\documentclass{article}

\usepackage[preprint]{neurips_2026}
\usepackage{comment}

\usepackage[utf8]{inputenc} 
\usepackage[T1]{fontenc}    
\usepackage{hyperref}       
\usepackage{url}            
\usepackage{booktabs}       
\usepackage{amsfonts}       
\usepackage{nicefrac}       
\usepackage{microtype}      
\usepackage{xcolor}         

\usepackage{amsmath, mathtools, amssymb}

\newcommand{\KL}{D_{\textnormal{KL}}}

\newcommand{\diag}{\textnormal{diag}}

\newcommand{\argmin}{\textnormal{argmin}}
\DeclarePairedDelimiterX{\inp}[2]{\langle}{\rangle}{#1, #2}
\DeclarePairedDelimiterX{\norm}[1]{\lVert}{\rVert}{#1}

\newcommand{\bfB}{\mathbf{B}}
\newcommand{\bfb}{\mathbf{b}}
\newcommand{\bfd}{\mathbf{d}}
\newcommand{\bfp}{\mathbf{p}}
\newcommand{\bfq}{\mathbf{q}}
\newcommand{\bfs}{\mathbf{s}}

\newcommand{\bfx}{\mathbf{x}}
\newcommand{\bfy}{\mathbf{y}}

\newcommand{\RR}{\mathbb{R}}

\usepackage{amsthm}
\newtheorem{theorem}{Theorem}
\newtheorem{corollary}{Corollary}
\newtheorem{proposition}{Proposition}
\newtheorem{lemma}{Lemma}
\theoremstyle{definition}
\newtheorem{definition}{Definition}
\newtheorem{assumption}{Assumption}
\newtheorem{remark}{Remark}

\newenvironment{repeatproposition}[1]
  {\renewcommand{\theproposition}{#1}\begin{proposition}}
  {\end{proposition}\addtocounter{proposition}{-1}}
\newenvironment{repeatlemma}[1]
  {\renewcommand{\thelemma}{#1}\begin{lemma}}
  {\end{lemma}\addtocounter{lemma}{-1}}
\newenvironment{repeattheorem}[1]
  {\renewcommand{\thetheorem}{#1}\begin{theorem}}
  {\end{theorem}\addtocounter{theorem}{-1}}

\usepackage{hyperref}
\usepackage[capitalize]{cleveref}
\crefname{assumption}{Assumption}{Assumptions}

\usepackage{xcolor}

\usepackage{natbib}

\title{Fully Distributed T\^atonnement for Chores Markets}

\author{%
  Bhaskar Ray Chaudhury \\
  University of Illinois Urbana-Champaign \\
  \And
  Christian Kroer \\
  Columbia University \\
  \And
  Ruta Mehta \\
  University of Illinois Urbana-Champaign \\
  \And
  Tianlong Nan \\
  Columbia University \\
}

\begin{document}

\maketitle

\begin{abstract}
We study price-adjustment dynamics for computing competitive equilibria (CE) in Fisher markets with chores. Unlike in classical goods markets, prices in chores markets are payments for taking on undesirable tasks, and natural excess-demand dynamics can fail; even the naïve analogue of Walrasian tâtonnement may diverge. Recent work of~\cite{chaudhury2025t} overcomes this obstacle via \emph{relative tâtonnement}, which subtracts the average excess-demand signal from the excess demand vector. This recovers convergence, but at the cost of coupling the price updates across all chores. This leaves open whether such global coupling is inherent, or whether convergent tâtonnement can be recovered through a genuinely local update in which each chore reacts only to its own excess demand.
    
We answer this question affirmatively through \emph{multiplicative tâtonnement}, a fully distributed dynamics in which each chore price is updated using only its current price and its own excess-demand signal. Although the update contains no explicit normalization term, Walras' law and the multiplicative form of the update implicitly preserve the relevant aggregate price geometry. We prove that multiplicative tâtonnement converges to a CE in any chores Fisher market with continuous, convex, and $1$-homogeneous (CCH) disutilities.
For convex CES disutilities, we further prove an approximate-CE convergence rate with the same $O(1/\varepsilon^2)$ dependence as relative tâtonnement, but with improved dependence on problem constants.
Experiments on real-world and simulated instances show that multiplicative tâtonnement is substantially faster in practice, often by an order of magnitude.

\end{abstract}

\section{Introduction}

Competitive equilibrium (CE) is one of the central solution concepts in market design and microeconomic theory. The \emph{Fisher market} is one of the most fundamental market settings.
In a classical Fisher market with \emph{goods}, there is a set of $m$ goods and $n$ agents. Each agent is endowed with a fixed budget of money, and once prices are assigned to the goods, agents spend their budgets to obtain utility-maximizing bundles. Prices are at equilibrium when these individual demands exactly clear the market. Beyond the static existence and welfare properties of CE, a classical question asks whether natural market dynamics can find such prices. Walrasian tâtonnement is the
canonical example: prices are adjusted in response to excess demand, increasing when demand exceeds supply and decreasing when supply exceeds demand~\citep{walras1874elements}. The stability and convergence of such dynamics have been studied since the foundational work of~\cite{samuelson1941stability,arrow1958stability,arrow1959stability}, and remain active topics in algorithmic game theory~\cite{codenotti2005market,cole2008fast,cheung2019tatonnement,goktas2023tatonnement,nan2025convergence}.

This paper studies tâtonnement dynamics in Fisher markets with \emph{chores}. In a chores market, agents do not pay to receive goods; rather, they are paid to perform undesirable tasks. Thus prices represent payments per unit of work, and the natural response to imbalance is reversed: when a chore is under-demanded, its price should increase to attract agents, while an over-demanded chore should become less lucrative. This reversal leads to a substantially different dynamical landscape. Chores markets can have disconnected sets of equilibria, and excess demand no longer has the monotonic structure that supports many goods-market analyses. 
Correspondingly, the computational literature on chores and mixed manna deals with non-convex optimization issues and is more nuanced than in the goods case~\cite{bogomolnaia2017competitive,garg2020computing,BoodaghiansCM22,ChaudhuryGMM22,BranzeiS24,chaudhury2024competitive}.

\textbf{Relative Tâtonnement.}
The recent work of~\cite{chaudhury2025t} initiated the study of tâtonnement dynamics for chores markets. They show that the naïve additive analogue of tâtonnement can diverge, even in simple chores markets, and propose \emph{relative tâtonnement}, which restores convergence by subtracting the average excess-demand signal from each coordinate. Relative tâtonnement shows that it is indeed possible to recover convergence for a t\^atonnement-like procedure, but at the cost of a global
normalization to excess demand. \textcolor{black}{Crucially, this normalization requires each price to be updated with knowledge of the 
excess-demand signals of all other chores, making relative t\^atonnement 
not fully distributed: it is a centrally coordinated adjustment process rather 
than a truly decentralized one.}

\textbf{The Appeal of Decoupled Price-Adjustment Dynamics.} 
\textcolor{black}{Economically, the tatonnement process simulates how prices adjust in \emph{real markets}: each good has a manager, auctioneer, or a market-maker who observes only the imbalance between supply and demand for that good, and adjusts its price accordingly. This decentralized vision is precisely what makes tatonnement a compelling model of market dynamics — it does not require a central planner who aggregates information across all goods, nor does it require any individual agent to know the global state of the market. Each price responds only to its own local signal, and equilibrium emerges as a \emph{collective consequence} of many independent local adjustments. This is, in Hayek's terms~\cite{hayek1945use}, the informational miracle of the price system: global co-ordination is achieved through purely local price updates.}

Our paper asks whether there exist truly decoupled t\^atonnement dynamics for the chores setting?
We answer this question in the affirmative by giving a multiplicative update in which each chore reacts only to its own excess demand signal.
Our dynamics, which we call \emph{multiplicative tâtonnement}, updates each chore price according to
\[
    p_j^{k+1} = p_j^k(1 + \eta^k y_j^k),
\]
where $y_j^k$ is the excess supply of chore $j$ at the current prices $p_j^k$s. The update is fully local: the manager of chore $j$ only needs to know the current payment $p_j^k$ and the imbalance signal $y_j^k$ for that chore. No average excess-demand term, price-simplex projection, or information about other chores is used. The key observation is that the correction imposed explicitly by relative tâtonnement emerges implicitly from the multiplicative update. When prices are initialized in the budget simplex, Walras' law makes the multiplicative dynamics preserve the relevant aggregate price geometry while still allowing each coordinate to update independently.

\subsection{Our Contributions}
We make four main contributions. 
\begin{enumerate}
    \item First, we introduce multiplicative tâtonnement as a fully distributed market dynamics for chores Fisher markets.
    
    \item Second, we prove that this dynamics converges to a CE for general continuous, convex, and $1$-homogeneous (CCH) disutilities (Theorem~\ref{thm:asymptotic-convergence}). The proof relates the update to entropic mirror descent, while controlling the perturbation created by using the simple multiplicative rule rather than the exact exponential update that is induced by entropic mirror descent.
    
    \item Third, for convex CES disutilities, we prove convergence rates to approximate CE using a relative-smoothness argument on the price simplex (Theorem~\ref{thm:convergence-rate}). This yields the same $O(1/\varepsilon^2)$ dependence on the approximation parameter $\epsilon$ as relative tâtonnement, with improved dependence on problem constants, while avoiding any coupling.

    \item Finally, we compare the dynamics empirically on real-world and simulated instances, and find that multiplicative tâtonnement substantially outperforms relative tâtonnement in practice (Section~\ref{sec:numerical-experiments}).
\end{enumerate}

    



\subsection{Related work}

\noindent\textbf{Economic dynamics in goods markets.}
The study of economic dynamics in goods markets has been prevalent in the theoretical computer science community since the 2000s, often with a focus on algorithmic complexity. 
The first polynomial-time algorithm for linear Fisher markets, introduced by~\cite{devanur2008market}, uses a primal-dual approach that can be interpreted as an economic dynamics. 
In the more general setting of exchange economies satisfying weak gross substitutability (WGS), \cite{codenotti2005market} showed that discrete-time additive t\^atonnement dynamics converges to an approximate equilibrium. 
\cite{cole2008fast,cheung2012tatonnement} showed that multiplicative t\^atonnement dynamics with artificial upper bounds on excess demand converges in certain nonlinear \emph{constant elasticity of substitution} (CES) markets. 
Along this line of research, \cite{cheung2019tatonnement,cole2019balancing,goktas2023tatonnement,nan2025convergence} further improved convergence rates for a wide range of markets.

\noindent\textbf{Economic dynamics and optimization.}
T\^atonnement dynamics is often connected to first-order methods through the well-known Eisenberg--Gale convex program~\cite{eisenberg1959consensus,goktas2021consumer,cheung2013tatonnement}. 
In~\cite{nan2023fast}, the authors show that even classic proximal gradient descent can have an economic-dynamics interpretation. 
There is also recent work on price-adjustment dynamics related to other optimization methods~\cite{zhang2025second,chen2025accelerated}. 
Proportional response (PR) dynamics is another popular class of economic dynamics, known to converge to a CE for CES utilities in goods markets~\citep{wu2007proportional,zhang2011proportional}. 
\cite{birnbaum2011distributed} shows that PR dynamics is equivalent to mirror descent on the convex Shmyrev program~\cite{shmyrev2009algorithm}. 
Better convergence rates and generalized versions of PR dynamics have also been established~\cite{cheung2018dynamics,cheung2025proportional}. 
For a detailed survey of proportional response dynamics, we refer the reader to~\cite{branzei2021proportional}. 
Recent work~\citep{tao2025fisher} also studies market dynamics, including t\^atonnement and proportional response dynamics, for public goods.

\noindent\textbf{CE in chores markets.}
In contrast to goods markets, computing a CE in chores markets is known to be harder~\cite{ChaudhuryGMM22}. 
This difficulty stems from the ``poles'' in the classic Eisenberg-Gale (EG) program for chores markets, and economic dynamics inherits this difficulty~\cite{bogomolnaia2017competitive,BoodaghiansCM22,chaudhury2025t}. 
\cite{chaudhury2024competitive} opened up a direction for developing useful optimization algorithms by introducing a dual EG program without the ``pole'' issue. 
Later, a general version of this dual EG program was discovered for continuous, convex, $1$-homogeneous chores markets~\cite{chaudhury2025t,tao2025fisher}.
\cite{chaudhury2025t} prove a rate of $\tilde{\mathcal{O}}(1/\varepsilon^2)$ for computing an approximate CE in a wide range of CES chores markets.

\section{Preliminaries}
\label{sec:pre}

\paragraph{Chores Fisher Markets.}
A chores Fisher market consists of a set of $n$ agents and a set of $m$ divisible chores.
Agents are deciding on what chores to take on from in order to earn money.
Each agent $i$ incurs a disutility $d_i (\bfx_i)$ for a bundle $\bfx_i = (x_{i1}, x_{i2}, \dots, x_{im})$ of chores that is assigned to her, where 
$x_{ij}$ represents the amount of chore $j$ allocated to agent $i$.
Each agent $i$ has an earning requirement of $B_i > 0$.
We consider each $d_i: \mathbb{R}^m_+ \rightarrow \mathbb{R}_+$ to be a general \emph{continuous, convex, $1$-homogeneous} (CCH) disutility function, and $d_i(\bfx_i) > 0$ for any nonzero $\bfx_i$.
We call such markets as \emph{CCH chores Fisher markets}.

Each chore $j$ has a \emph{fixed} supply $s_j > 0$. 
An allocation is denoted by a matrix 
$\bfx \in \mathbb{R}^{n \times m}_{+}$ whose $i$th row is $\bfx_i$. 
Given prices $\bfp = (p_1, \dots, p_m)$ of chores, 
where $p_j$ represents the payment-per-unit of chore $j$ done, agent $i$ will demand a bundle of chores that minimizes her disutility subject to satisfying her earning requirements, i.e., the demand of agent $i$ is denoted by 
\begin{equation*}
    X_i(\bfp) := \textnormal{argmin}_{\bfx'_i \in \mathbb{R}^m_+} \left\{  d_i(\bfx'_i) \mid \inp*{\bfp}{\bfx'_i} \geq B_i \right.\}.
\end{equation*}
A competitive equilibrium of a chores Fisher market is defined as follows.
\begin{definition}[Competitive equilibrium]
    A pair of price $\bfp^* \in \RR^m_+$ and allocation $\bfx^* \in \RR^{n \times m}_+$ is a competitive equilibrium (CE) if 
    \begin{itemize}
        \item (optimal bundles) $\bfx^*_i \in X_i(\bfp^*)$ for each agent $i$, and  
        \item (market clearing) $\sum_{i = 1}^n x^*_{ij} \geq s_j$ and $p_j (\sum_{i = 1}^n x^*_{ij} - s_j) = 0$ for each chore $j$.
    \end{itemize}
\end{definition}
Without loss of generality, we make the following assumptions.
\begin{assumption}
    (1) $s_j = 1$ for any chore $j$.
    (2) $\sum_{i=1}^n B_i = 1$.
    (3) $d_i(\mathbf{1}_m) = 1$ for any agent $i$.
    \label{assump:wlog}
\end{assumption}
We show in~\cref{app:sec:chores-FM} that any CCH chores Fisher market  can be scaled to satisfy the above assumption, and that this scaling induces a bijection between the sets of CE of the two markets.
We further make the following natural assumption on disutility functions.
\begin{assumption}
    For each agent $i$, her disutility funciton $d_i(\cdot)$ is strictly increasing in its domain in each coordinate.
    \label{assump:strictly-increasing}
\end{assumption}
Under~\cref{assump:strictly-increasing}, no agent would do a zero-priced task, and therefore, it must be that $p^*_j>0$ for all $j\in [m]$ at equilibrium to ensure that every chore is done. As a result, the market clearing condition reduces to $\sum_{i = 1}^n x_{ij} = 1$ for all $j \in [m]$.
In this paper, we consider the following approximate CE~\cite{chaudhury2025t}.
\begin{definition}[$\varepsilon$-approximate CE]
A price vector $\bfp \in \mathbb{R}^m_+$ and an allocation $\bfx \in \mathbb{R}^{n \times m}_{+}$ satisfy $\varepsilon$-competitive equilibrium (CE) if and only if:
\begin{enumerate}
\item $\langle \bfp, \bfx_i \rangle 
= B_i$ for all $i \in [n]$;
\item $d_i(\bfx_i) \leq d_i(\bfy_i)$ for all $\bfy_i$ such that $\langle \bfp, \bfy_i \rangle \geq \langle \bfp, \bfx_i \rangle,\ \forall i \in [n]$;
\item $1 - \varepsilon \leq \sum_{i = 1}^n x_{ij} \leq \frac{1}{1 - \varepsilon}$ for all $j \in [m]$.
\end{enumerate}
\label{def:approx-CE-2}
\end{definition}

That is, only the market clearing condition is approximately satisfied, the rest of the equilibrium conditions are satisfied exactly.
A common class of CCH disutility functions is the class of convex CES disutilities: 
\begin{equation}
    d_i(\bfx_i) = \Big( \sum\nolimits_{j = 1}^m (d_{ij} x_{ij})^\rho \Big)^{\frac{1}{\rho}}, \quad d_{ij} > 0 \ \forall \ j \in [m], \, \rho \in (1, \infty).
    \tag{Convex CES Disutilities}
\end{equation}
Convex CES disutilities are CCH, strictly positive for any nonzero $\bfx_i$, and satisfy~\cref{assump:strictly-increasing}.

\paragraph{Excess Supply and T\^atonnement Dynamics.}
For a given price vector $\bfp \in \RR^m_+$, the excess supply mapping is defined as \emph{(excess-demand is negative of excess-supply)}
\begin{equation*}
    Y(\bfp) = \Big\{ \mathbf{1}_m - \sum\nolimits_{i = 1}^n \bfx_i \ \Big| \ \bfx_i \in X_i(\bfp) \Big\}.
\end{equation*}
We call any $\bfy \in Y(\bfp)$ an excess supply (vector).
If there exists a $\bfy \in Y(\bfp)$ such that $\bfy = \mathbf{0}_m$, then $\bfp$ corresponds to a CE.
If there exists a $\bfy \in Y(\bfp)$ such that $\norm{\bfy}_2 \leq \varepsilon$, then $\bfp$ corresponds to an $\varepsilon$-approximate CE defined in~\cref{def:approx-CE-2}; 
we add a proof these facts in~\cref{app:sec:chores-FM} for completeness.

T\^atonnement dynamics is a price-adjustment process driven by excess supply.
Intuitively, in a chores Fisher market, if the supply of a chore exceeds the demand for it, then the price of the chore should be increased to attract more agents to take on the chore. 
Conversely, if the supply is below the demand, then the price should be decreased so that agents become less willing to perform the chore.
\cite{chaudhury2025t} considered two types of t\^atonnement dynamics: 
\begin{equation}
    \begin{aligned}
        \bfy^k &\in Y(\bfp^k) \\ 
        \bfp^{k + 1} &= \bfp^k + \eta^k \bfy^k, 
    \end{aligned}
    \tag{Na\"ive T\^atonnement}
    \label{alg:naive-ttm}
\end{equation}
and 
\begin{equation}
    \begin{aligned}
        \bfy^k &\in Y(\bfp^k) \\ 
        \tilde{\bfy}^k &= \bfy^k - \big( \frac{1}{m} \mathbf{1}_m^\top \bfy^k \big) \cdot \mathbf{1}_m \\ 
        \bfp^{k + 1} &= \bfp^k + \eta^k \tilde{\bfy}^k.
    \end{aligned}
    \tag{Relative T\^atonnement}
    \label{alg:relative-ttm}
\end{equation}
They show that, \eqref{alg:naive-ttm} dynamics can diverge in chores Fisher markets with convex CES disutilities, whereas \eqref{alg:relative-ttm} dynamics converges in any CCH chores Fisher market satisfying~\cref{assump:strictly-increasing}.

\paragraph{Mirror descent.}
Let $h: \mathbb{R}^m \rightarrow \mathbb{R}$ be a function that is: a) strictly convex, b) continuously differentiable, c) defined on a closed convex set.
Then, the Bregman divergence is defined as 
\begin{equation*}
    D_h(\bfp, \bfq) = h(\bfp) - h(\bfq) - \langle \nabla h(\bfq), \bfp - \bfq \rangle.
\end{equation*}
If we takes $h_{\textnormal{KL}}(\bfp) = \sum_{j = 1}^m p_j \log{p_j}$, $\KL(\bfp, \bfq) = \sum_{j = 1}^m p_j \log{\frac{p_j}{q_j}} - p_j + q_j$, which is classic relative entropy, or KL divergence.
Mirror descent is a first-order optimization method that uses a Bregman divergence as its proximal geometry. 
Given a differentiable objective $f: \mathcal{P} \rightarrow \mathbb{R}$ and $\mathcal{P} \subset \mathbb{R}^m$, mirror descent updates 
\begin{equation*}
    \bfp^{k + 1} = \argmin_{\bfp \in \mathcal{P}} \left\{ \eta^k \inp{\nabla f(\bfp^k)}{\bfp} + D_h(\bfp, \bfp^k) \right\}.
\end{equation*}
Relative smoothness generalizes the usual Lipschitz-gradient smoothness condition. 
We say that $f$ is $L$-smooth relative to $h$ if
\begin{equation*}
    f(\bfq) \leq f(\bfp) + \inp{\nabla f(\bfp)}{\bfq - \bfp} + L D_h(\bfq, \bfp), \quad \textnormal{ for all } \bfp, \bfq \in \mathcal{P}.
\end{equation*}

\section{Multiplicative t\^atonnement and its convergence}
\label{sec:multiplicative-tatonnement-convergence}

\subsection{Multiplicative t\^atonnement dynamics}

We begin by introducing a new dynamics for chores Fisher market.
For each chore $j$, suppose there is a manager $j$ responsible for allocating that chore among the agents.
In each round $k$, manager $j$ submits a payment-per-unit $p^k_j$ to the market and receive an excess-supply signal $y^k_j$.
She then updates the payment-per-unit multiplicatively in response to this signal: if there is positive excess supply (more of the chore is supplied than agents want to perform), she increases the payment, and vice versa.
The update in round $k$ can be written compactly as follows, where $\odot$ denotes coordinate-wise product, and $\eta^k$ is the step size:  
\begin{equation}
    \begin{aligned}
        \bfy^k &\in Y(\bfp^k) \\ 
        \bfp^{k + 1} &= \bfp^k \odot (\mathbf{1}_m + \eta^k \bfy^k).
    \end{aligned}
    \tag{Multiplicative T\^atonnement}
    \label{alg:mul-ttm}
\end{equation}

This multiplicative dynamics has two key advantages:
\begin{itemize}
    \item It is \emph{fully distributed}: each manager observes only the feedback associated with her own chore, and not the feedback associated with other chores.
    \item It scales the adjustment by both the excess-supply signal and the current payment-per-unit.
\end{itemize}

\subsection{Asymptotic convergence of multiplicative t\^atonnement}

Next, we prove that, \eqref{alg:mul-ttm} dynamics converges to a CE in any CCH chores Fisher market.
This fact is in contrast to its additive analogue as Na\"ive T\^atonnement is known to provably diverge across a family of CCH chores Fisher markets~\citep{chaudhury2025t}.

One major obstacle to guaranteeing convergence of dynamics to a CE in chores markets is the divergent behavior of prices once they leave the ``price simplex'', as pointed out by~\cite{chaudhury2025t}.
In contrast to approaches that impose a correction on the excess-supply signal, multiplicative t\^atonnement prevents price divergence by exploiting Walras’ law, a fundamental identity in general equilibrium theory.
All proofs in this section are deferred to~\cref{app:sec:convergence}.
\begin{proposition}[Walras's law]
    If $\bfp \in \Delta_B$, then we have 
    $\inp{\bfp}{\bfy} = 0$ for all $\bfy \in Y(\bfp)$. 
    \label{prop:sum-p-zeta-equal-0}
\end{proposition}

By leveraging this intrinsic property of excess supply, 
the price vectors generated by multiplicative tâtonnement automatically stay within the affine hull of $\Delta_B$, if the initial price vector lies in $\Delta_B$.
Beyond this fact, to ensure convergence of the dynamics, we additionally require boundedness of the excess supply and strict positivity of the prices.


\begin{lemma}[Boundedness of excess supply]
    In a CCH chores Fisher market that satisfies~\cref{assump:strictly-increasing}, denote $d_{\min}^{(i)} := \min_{j \in [m]} d_i(\mathbf{e}_j) > 0$ where $\mathbf{e}_j$ is the basis vector with a $1$ in coordinate $j$ and $0$ elsewhere. 
    Then, for any $\bfp \in \Delta_\bfB$, 
    $1 - (d_{\min})^{-1} \leq y_{j} \leq 1$ for any $j \in [m]$ and $\bfy \in Y(\bfp)$, where $d_{\min} := \min_i d_{\min}^{(i)}$.
    If $d_i(\cdot)$ is a convex CES disutility function for each agent $i$, then $1 - (\min_{i, j} d_{ij})^{-1} \leq y_{j} \leq 1$ for any $j \in [m]$ and $\bfy \in Y(\bfp)$.
    \label{lem:eta-y-bound}
\end{lemma}

Define $\ell_0 = \frac{1}{3m} \min_{i} \big\{ \min \{{d'_{ij}}/{d'_{ij'}}\ |\ {j, j' \in [m]},  \tilde{x}_{ij} \geq  \frac{1}{2 n m}, \tilde{x}_{ij'} \leq 2m B_i,  \bfd'_i \in \partial d_i(\tilde{\bfx}_i) \} \big\} > 0$.
\begin{lemma}[Strict positiveness of prices]
    Let $\beta := 1 + (d_{\min})^{-1} > 0$.
    Let $\{ \bfp^k \}_{k \geq 0}$ be a sequence of iterates generated by~\eqref{alg:mul-ttm} with the initial point $\bfp^0 \in \textnormal{relint}\left( \Delta_B \right)$ and $\eta^k \leq \frac{1}{2\beta}$ for all $k$. 
    Then, we have 
    $(i)$ $\bfp^k \in \Delta_B$ for all $k \geq 0$, 
    $(ii)$ $p^{k+1}_j > p^k_j + \frac{1}{2}\eta^k$ if $p^k_j \leq \ell_0$ for any $k \geq 0$ and $j \in [m]$, 
    $(iii)$ if additionally $\sum_{k=0}^\infty \eta^k = \infty$, then there exists a finite index $k_0 \geq 0$ such that $p^k_j \geq \frac{\ell_0}{2}\; \forall\, j \in [m]$ for all $k \geq k_0$.
    \label{lem:discrete-time-multiplicative-tatonnement-properties}
\end{lemma}

We then prove the convergence of multiplicative t\^atonnement by linking it to the following updates:
\begin{equation}
    \bar{\bfp}^{k + 1} := \bfp^k \odot \exp{(\eta \bfy^k)},
    \label{eq:entropy-ttm}
\end{equation}
which is mirror descent w.r.t. the relative entropy for the following objective function~\citep[Lemma 6]{chaudhury2025t}
\begin{equation}
    f(\bfp) := - \displaystyle\sum_{j = 1}^m p_j + \sum_{i = 1}^n B_i \log\left( \max_{\bfx_i \geq 0: d_i(\bfx_i) \leq 1} \inp{\bfp}{\bfx_i} \right).
    \label{def:function-f}
\end{equation}
Note that \eqref{eq:entropy-ttm} itself does not converge since the function $f$ is not globally lower bounded.
In particular, the objective value goes to $-\infty$ as $\bfp \rightarrow \mathbf{0}_m$, or as $p_j \rightarrow \infty$ for any $j \in [m]$.
However, it serves as a useful reference process for proving the per-step improvement: 
multiplicative t\^atonnement tracks \eqref{eq:entropy-ttm} within the price simplex $\Delta_B$, and thus inherits the iterate descent property up to a small perturbation.
Since $f$ is lower bounded within $\Delta_B$, the iterate converges to a stationary point asymptotically.


Formally, \eqref{alg:mul-ttm} dynamics fits into the following Bregman subgradient update scheme~\citep{ding2024stochastic}: 
\begin{equation}
    \bfp^{k + 1} \approx \argmin_{\bfq} \{ \eta^k \inp{-\bfy^k}{\bfq} + \KL(\bfq \Vert \bfp^k) \}, 
    \label{eq:bregman-subgradient-update-1}
\end{equation}
and the perturbation error is controlled by 
\begin{equation}
    \norm*{\frac{\nabla h_{\textnormal{KL}}(\bfp^{k + 1}) - \nabla h_{\textnormal{KL}}(\bfp^k)}{\eta^k} - \bfy^k}_2 \leq \nu^k.
    \label{eq:bregman-subgradient-update-2}
\end{equation}

\cite{ding2024stochastic} prove the convergence of the above process under a set of regularity conditions. 
We adapt their result to our setting and state it as the following proposition.
\begin{proposition}
    Let $\{ \bfp^k \}_{k \geq 0}$ be a sequence of iterates generated by Bregman subgradient update described in~\cref{eq:bregman-subgradient-update-1,eq:bregman-subgradient-update-2}, and suppose that the process satisfies that 
    $(1)$ the sequences $\{ \bfp^k \}_{k \geq 0}$, $\{ \log{(\bfp^k)} \}_{k \geq 0}$, and $\{ \bfy^k \}_{k \geq 0}$ are uniformly bounded almost surely,
    $(2)$ the stepsizes satisfies $\sum_{k=0}^\infty \eta^k = \infty$, and either $\eta^k = o(\frac{1}{\log{k}})$ or $\sum_{k=0}^\infty (\eta^k)^2 < \infty$, 
    $(3)$ the perturbation error satisfies $\lim_{k \rightarrow \infty} \nu^k \rightarrow 0$, 
    $(4)$ there is a potential function $f$ such that $-\bfy^k \in \partial f(\bfp^k)$\footnote{$\partial f(\bfp^k)$ denotes Clarke differential of function $f$ at $\bfp^k$.}, and $f$ is lower bounded, 
    $(5)$ the critical value set $\{ f(\bfp) \mid \mathbf{0}_m \in \partial f(\bfp) \}$ has empty interior in $\mathbb{R}$. 
    Then, almost surely, any cluster point of $\{ \bfp^k \}_{k \geq 0}$ is a critical point and the function values $\{ f(\bfp^k) \}_{k \geq 0}$ converge.
    \label{prop:ding-proposition}
\end{proposition}

Therefore, by combining~\cref{lem:eta-y-bound,lem:discrete-time-multiplicative-tatonnement-properties}, and choosing the proper stepsizes to control the perturbation error, we can prove the convergence of~\eqref{alg:mul-ttm}.
\begin{theorem}[Asymptotic convergence]
    Let $\{ \eta^k \}_{k \geq 0}$ be a sequence of stepsizes satisfying (i) $0 < \eta^k \leq \frac{1}{2\beta}$, (ii) $\sum_{k = 0}^\infty \eta^k = \infty$, and (iii) $\eta^k = o(\frac{1}{\log{k}})$ or $\sum_{k=0}^\infty (\eta^k)^2 < \infty$.
    Let $\{ \bfp^k \}_{k \geq 0}$ be a sequence of iterates generated by~\eqref{alg:mul-ttm} with $\{ \eta^k \}_{k \geq 0}$ and any initial point $\bfp^0 \in \textnormal{relint}(\Delta_B)$. 
    Then, every limit point of $\{ \bfp^k \}_{k\geq 0}$ is a CE of the chores Fisher market.
    \label{thm:asymptotic-convergence}
\end{theorem}

\begin{remark}
    Notice that the modulus $\ell_0$ can be very small in the worst case.
    For example, convex CES disutilities with a large $\rho$ cause $\ell_0 = O((D / m)^{\rho - 1})$ where $D = \max_i (\max_j d_{ij}/\min_j d_{ij})$.
    This can be problematic for relative t\^atonnement, since its stepsize must scale with $\ell_0$ to keep prices strictly positive\footnote{We do not consider price projection, since it would require each chore manager to have access to the prices of other chores.}.
    In contrast, multiplicative t\^atonnement dynamics admits stepsizes independent of $\ell_0$, since its updates scale multiplicatively with the prices.
\end{remark}

\section{Convergence rates to an approximate CE under convex CES disutilities}
\label{sec:convergence-rate}

We further study convergence rates of multiplicative tâtonnement to find an $\varepsilon$-approximate CE.
We focus on the class of convex CES disutility function, 
i.e., $d_i(\bfx_i) = (\sum_{j = 1}^m (d_{ij} x_{ij})^\rho)^{1/\rho}$ with $\rho \in (1, \infty)$.
Since the inner-log maximization problem in~\cref{def:function-f} ensures its value bounded away from zero, and has a unique maximizer, function $f$ is differentiable: 
\begin{equation*}
    f(\bfp) = - \sum_{j = 1}^m p_j + \sum_{i = 1}^n \frac{B_i}{\sigma} \log{\Big( \sum_{j = 1}^m d_{ij}^{-\sigma} p_j^\sigma \Big)}, \quad \sigma = \frac{\rho}{\rho - 1} \in (1, \infty).
\end{equation*}

To study the convergence rate of~\eqref{alg:relative-ttm} dynamics, analyze its smoothness with respect to the squared Euclidean norm, which is the standard notion of smoothness.
The convergence rate of~\eqref{alg:mul-ttm} dynamics, by contrast, is governed by relative smoothness with respect to $h_{\textnormal{KL}}$.
By leveraging the fact that $\bfp \in \Delta_B$, we obtain an improved relative smoothness than the standard one.
In particular, it establishes the smoothness over entire $\Delta_B$ for the full spectrum of $\rho \in (1, \infty)$, and the modulus does \emph{not depend} on a lower bound on the prices.
This relative smoothness then lays the foundation for a tighter convergence rate.
All proofs in this section is in~\cref{app:sec:convergence-rate}.

\begin{lemma}[Relative smoothness of $f$ w.r.t. $h_{\mathrm{KL}}$]
    For any two price vectors $\bfp, \bfq \in \Delta_B$, we have 
    \begin{equation*}
        f(\bfp) \leq f(\bfq) + \langle \nabla f(\bfq), \bfp - \bfq \rangle + (\sigma - 1)m^{\sigma - 1}(\min_{i,j} d_{ij})^{-\sigma} \KL(\bfp, \bfq).
    \end{equation*}
    \label{lem:relative-smoothness-KL}
\end{lemma}
\begin{remark}
    Note that the na\"ive global KL-relative smoothness of $f$ still has an $O(1/\ell_0)$ smoothness modulus, which means that the constraint $\sum_{j=1}^m p_j = \norm{\bfB}_1$
    is necessary to eliminate the singularity caused by the logarithm. 
    However, even under this constraint, $f$ still has an $O(1/\ell_0)$ smoothness modulus with respect to $\frac{1}{2}\lVert \cdot \rVert^2$, due to the mismatch between the Euclidean geometry and the geometry induced by $h_{\textnormal{KL}}$.
\end{remark}

With this tool in hand, we next prove a convergence rate to compute an $\varepsilon$-approximate CE.
We have
\begin{equation*}
    \bfp^{k + 1} = \argmin_{\bfq \geq 0} \{ \eta^k \inp{-\bfy^k + \nu^k}{\bfq} + \KL(\bfq \Vert \bfp^k) \}
\end{equation*}
where the perturbation $\nu^k$ is defined as 
\begin{equation*}
    \nu^k := -\frac{\log{(1 + \eta^k \bfy^k)}}{\eta^k} + \bfy^k.
\end{equation*}

By standard mirror descent analysis and~\cref{lem:relative-smoothness-KL}, we have the following descent lemma: 
\begin{lemma}[Descent property]
    Let $\{ \bfp^k \}_{k \geq 0}$ be a sequence of iterates generated by~\eqref{alg:mul-ttm} dynamics with any positive stepsizes $\{ \eta^k \}_{k \geq 0}$ and any initial point $\bfp^0 \in \textnormal{relint}(\Delta_B)$.
    We have that 
    \begin{equation}
        f(\bfp^{k + 1}) \leq f(\bfp^k) - \frac{1}{\eta^k} \left( \KL(\bfp^{k + 1}, \bfp^k) \right) - \inp{\nu^k}{\bfp^{k+1} - \bfp^k} + L \KL(\bfp^{k + 1}, \bfp^k).
        \label{eq:descent-inequality}
    \end{equation}
    \label{lem:descent-inequality}
\end{lemma}

Define $\norm{\bfy}_{\bfp} := \sqrt{\sum_{j=1}^m p_j (y_j)^2}$, we can then link the last two terms in~\cref{eq:descent-inequality} to $\norm{\bfy^k}_{\bfp^k}^2$.
To do so, we need the stepsizes to be uniformly small to control $\norm{\eta^k \bfy^k}_\infty$.
\begin{lemma}
    If $\eta^k \leq \frac{1}{2\beta}$ for all $k \geq 0$, then we have $- \inp{\nu^k}{\bfp^{k+1} - \bfp^k} \leq \beta(\eta^k)^2 \norm{\bfy^k}_{\bfp^k}^2$and $\KL(\bfp^{k + 1}, \bfp^k) \geq \frac{1}{3} (\eta^k)^2 \norm{\bfy^k}_{\bfp^k}^2$.
    \label{lem:link-distances}
\end{lemma}
Combining~\cref{lem:descent-inequality,lem:link-distances} yields the convergence rate.
\begin{theorem}[Convergence rate guarantee]
    Let $\{ \bfp^k \}_{k \geq 0}$ be a sequence of iterates generated by~\eqref{alg:mul-ttm} with any initial point $\bfp^0 \in \textnormal{relint}(\Delta_B)$ and any positive stepsizes $\{ \eta^k \}_{k \geq 0}$ satisfying 
    \begin{equation*}
        \eta^k \equiv \eta \leq \frac{1}{12(L + \beta)}, 
    \end{equation*}
    then we can find an iterate $(\bfp, \bfy)$ satisfies $\norm{\bfy}_{\bfp} \leq \varepsilon$ in $\lceil
    {12(f(\bfp^0) - \underline{f})}/{\eta \varepsilon^2} \rceil$ iterations.
    Here, $L = (\sigma-1)\,m^{\sigma-1}(\min_{i,j} d_{ij})^{-\sigma}$ and $\beta = 1 + (\min_{i,j} d_{ij})^{-1}$.
    \label{thm:convergence-rate}
\end{theorem}

By taking $\varepsilon = \sqrt{\ell_0 / 2} \varepsilon$ and $\eta = {1}/{(12(L + B))}$, we directly have the following corollary.
\begin{corollary}
    \eqref{alg:mul-ttm} dynamics with $\eta = {1}/{(12(L + \beta))}$ can find an $\varepsilon$-approximate CE in $\lceil
    {288(L + \beta)(f(\bfp^0) - \underline{f})}/{\ell_0\varepsilon^2} \rceil$ iterations.
    \label{crl:convergence-rate-approx-CE}
\end{corollary}

\paragraph{Discussions.}
\emph{First}, \cref{thm:convergence-rate} shows that we can find a price vector for which $\sum_{j=1}^m p_j (y_j)^2 \leq \varepsilon$ in $\mathcal{O}(1/\varepsilon^2)$ iterations, without any dependence on $1/\ell_0$.
In other words, for any price coordinate that is not close to zero, an $\varepsilon$-small excess supply is guaranteed.
We emphasize that as $\rho \rightarrow \infty$, $1/\ell_0$ grows exponentially fast and can therefore be problematic.
Since $L + \beta = (\sigma - 1)m^{\sigma - 1} (\min_{i,j} d_{ij})^{-\sigma} + 1 + (\min_{i,j} d_{ij})^{-1}$,
and $\sigma \rightarrow 1$ as $\rho \rightarrow \infty$, this convergence rate ensures that multiplicative t\^atonnement finds a meaningful price vector even as $\rho \rightarrow \infty$.

\emph{Second}, as we detail in~\cref{app:subsec:comparison-relative-tatonnement}, relative t\^atonnement has a dependence of
$\mathcal{O}\big(
    \frac{(\sigma - 1)(\min_{i,j} d_{ij})^{-(\sigma+1)}}{\sqrt{m}\ell_0 \varepsilon^2}
\big)$,
even under an analysis sharper than the original one in~\cite{chaudhury2025t}.
In contrast, multiplicative t\^atonnement improves this rate by a factor of $(\min_{i,j} d_{ij})^{-1}$.
As we will see in~\cref{sec:numerical-experiments}, this factor dominates in most real-world instances and can be very large.

\emph{Third}, \cref{crl:convergence-rate-approx-CE} indicates that we can compute an approximate CE using a constant stepsize that is much larger than that allowed by relative t\^atonnement.
In~\cref{sec:numerical-experiments}, we show that this larger stepsize can substantially improve the performance of multiplicative t\^atonnement.




\section{Numerical experiments}
\label{sec:numerical-experiments}

We conduct numerical experiments on three types of datasets: Spliddit instances, AAMAS bidding instances, and synthetic instances. For each type, we test three convex CES disutilities with $\rho \in \{1.2, \ 2, \ 5\}$. We compare the number of iterations required to compute an $\varepsilon$-approximate equilibrium, where $\varepsilon \in \{0.1, \ 0.05, \ 0.01, \ 0.005, \ 0.001\}$.
The experiments were implemented on a personal desktop with Apple M1 chip.

\paragraph{Spliddit instances.}
\begin{figure}[t]
    \centering
    \includegraphics[width=1\linewidth]{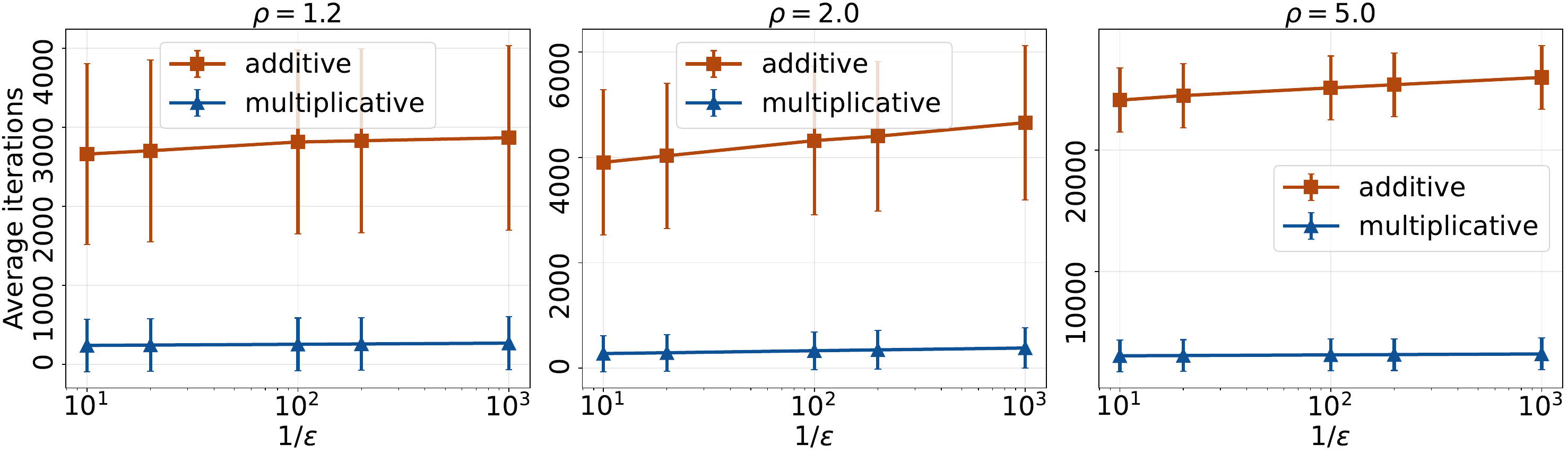}
    
    \caption{Multiplicative T\^atonnement v.s. Relative T\^atonnement on Spliddit instances.
    The marker denotes the mean value of the number of iterations required to compute an $\varepsilon$-approximate CE over all instances.
    The error bar corresponds to $0.1$ standard deviation.}
    \label{fig:spliddit}
\end{figure}
Spliddit is a research-driven online platform that helps users divide resources fairly and transparently~\citep{goldman2015spliddit}. We use a dataset containing $670$ task-distribution instances from Spliddit. The instance sizes range from $2$ to $6$ agents and from $2$ to $37$ tasks. We test all instances in the dataset and report the average number of iterations across different approximation tolerances $\varepsilon$ in~\cref{fig:spliddit}.

We use the following heuristic to select practical stepsizes for both algorithms. 
For each instance, we start with a large stepsize and iteratively search for an upper bound on the admissible stepsizes. 
If a trajectory leaves the price simplex, we set the current stepsize as an upper bound on the admissible range. 
We then perform a grid search (excluding the stepsize upper bound) over the admissible range using $9$ grid points, from the largest to the smallest. 
If all tested stepsizes are admissible, we choose the one that achieves the best performance, measured by the smallest approximation tolerance.
We also use this method to select practical stepsizes for other types of instances.

\paragraph{AAMAS bidding instances.}
We construct a set of AAMAS bidding instances from a dataset of bids submitted by PC members at the AAMAS conference for potential papers to review. This dataset is obtained from PrefLib~\citep{DBLP:conf/aldt/MatteiW13} and is also used by~\cite{chaudhury2024competitive}.
We convert reviewer bids into a disutility matrix using the mapping \texttt{yes:1}, \texttt{maybe:3}, \texttt{missing:5}, \texttt{no:7}, and \texttt{conflict:1401}. For each instance, we choose a random seed paper, take its nearest \texttt{100} papers under the squared distance between bidding profiles, and then select the \texttt{200} reviewers with the most positive responses on this paper set. Finally, we optionally add Gaussian noise to the resulting \texttt{100$\times$200} disutility matrix.
For each $\rho$, we sample \texttt{10} such instances and set all reviewers’ budgets to one, as is standard in fair-division settings. The results are shown in~\cref{fig:bidding}.

\paragraph{Synthetic instances.}
\begin{figure}[t]
    \centering
    \includegraphics[width=1\linewidth]{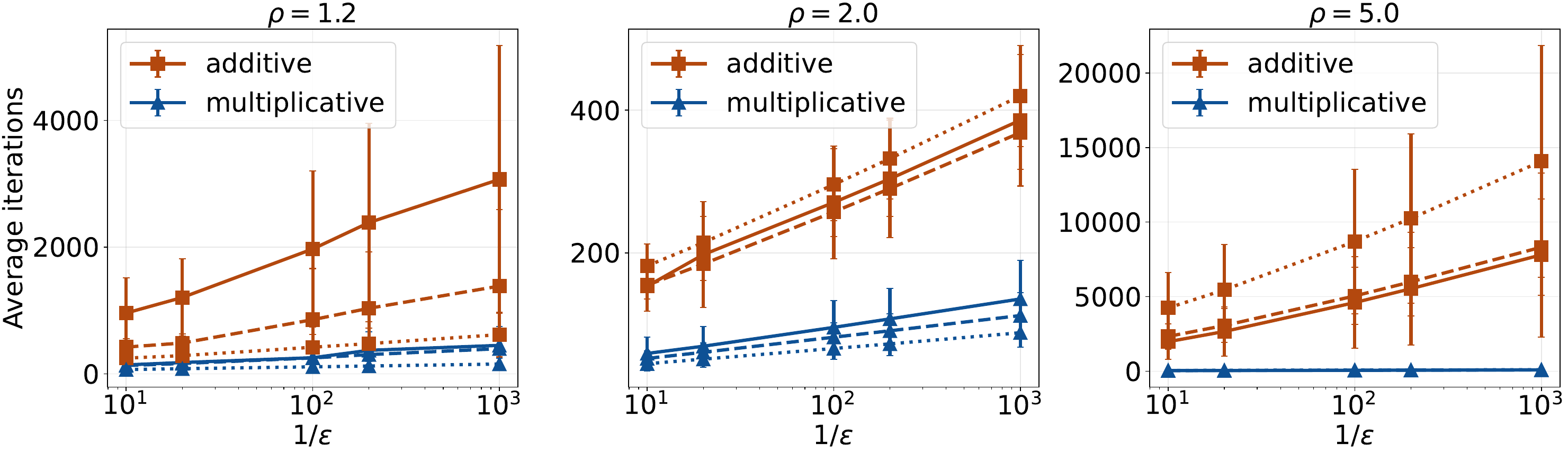}
    \caption{Multiplicative T\^atonnement v.s. Relative T\^atonnement on synthetic instances with truncated normal distributions.
    Solid lines correspond to $100 \times 500$ instances, dashed lines correspond to $100 \times 1000$ instances, dotted lines correspond to $100 \times 2000$ instances.
    The error bar corresponds to $0.1$ standard deviation.}
    \label{fig:synthetic-truncated_normal}
\end{figure}
We simulate a set of large instances to compare the empirical convergence rates of the two dynamics. We generate instances with $100$ agents and ${500,1000,2000}$ chores, reflecting the fact that, in many real applications, the number of chores is much larger than the number of agents. For each problem size, we test the dynamics on $20$ instances generated from $20$ random seeds.
We consider two valuation distributions: (1) a lognormal distribution with location parameter $0$ and scale parameter $1$; and (2) a truncated normal distribution with mean $0.5$ and standard deviation $0.2$, truncated below at $0.01$. We set each agent's budget to one. The results are shown in~\cref{fig:synthetic-truncated_normal,fig:synthetic-lognormal}.

\begin{figure}[h]
    \centering
    \includegraphics[width=1\linewidth]{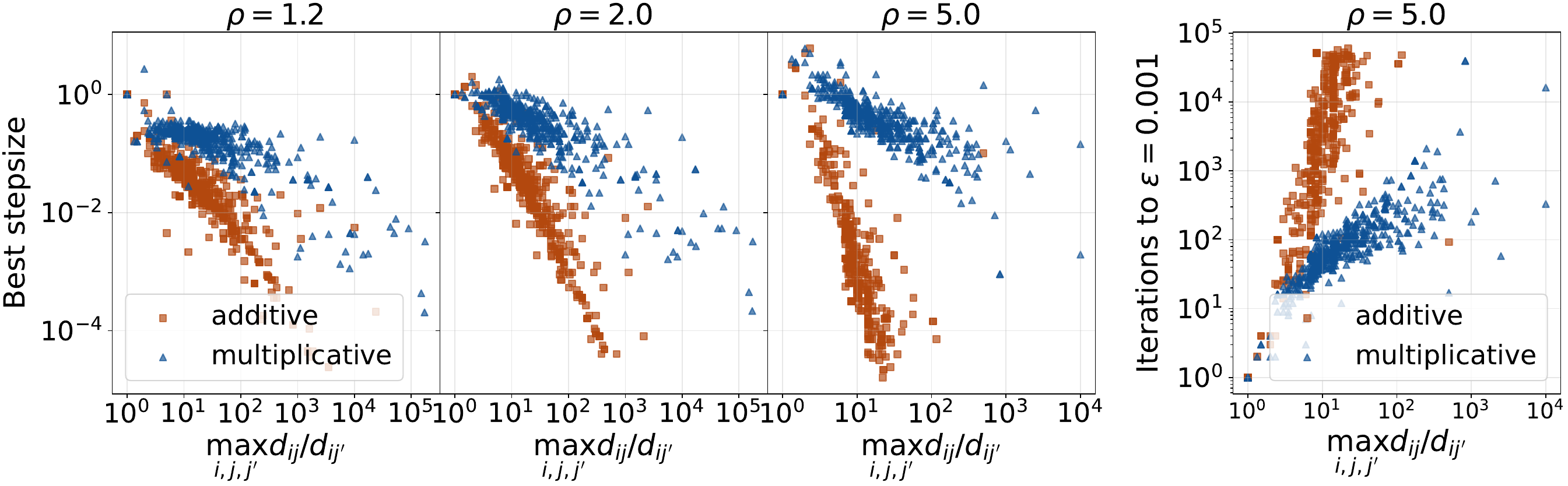}
    \caption{The best stepsizes in practice \emph{or} the number of iterations to compute an approximate CE v.s. maximum ratio of the disutility coefficients with $\rho \in \{ 1.2, \ 2, \ 5 \}$.
    Each point in the plots corresponding to a Spliddit instance that finds an approximate CE within the iteration budget.}
    \label{fig:max-value-ratio}
\end{figure}
\paragraph{Discussions.}
Overall, multiplicative t\^atonnement significantly outperforms relative t\^atonnement. 
On Spliddit instances, multiplicative t\^atonnement requires fewer than $100$ iterations to find a $0.001$-approximate CE, even when $\rho = 5$. In contrast, relative t\^atonnement may require a very large number of iterations to find an approximate CE, and this number increases as $\rho$ grows. In particular, relative t\^atonnement can take more than $20{,}000$ iterations to find a $0.1$-approximate CE.
On AMMAS bidding instances, although both dynamics compute approximate CE within $100$ iterations, multiplicative t\^atonnement still reduces the iteration cost by roughly half. These instances are relatively easy to solve because their value sets are sparse, containing only $5$ distinct values. On synthetic instances, multiplicative t\^atonnement typically achieves a \texttt{$10\times$} speedup, even for small values of $\rho$.

An interesting observation is the relationship between the best practical stepsizes and the maximum value ratio $\max_{i,j,j'} d_{ij}/d_{ij'}$. 
We illustrate this relationship based on the Spliddit instances with scatter plots in~\cref{fig:max-value-ratio}. 
The plots reveal a clear pattern: as $\max_{i,j,j'} d_{ij}/d_{ij'}$ increases, the best stepsize decreases at a rate consistent with our analysis. 
In particular, the largest admissible stepsize appears to be inversely proportional to $(\max_{i,j,j'} d_{ij}/d_{ij'})^\rho$.
This relationship becomes linear on a log-log plot, where the slope reflects the exponent. 
As shown in the three plots on the left, the slopes for relative t\^atonnement are approximately $1.2$, $2$, and $5$, respectively. 
In contrast, the slope for multiplicative t\^atonnement changes only mildly.

Additionally, from the right plot in~\cref{fig:max-value-ratio}, 
as $\max_{i, j, j'} d_{ij} / d_{ij'}$ increases, 
the increase of the number of iterations appearantly has a separation between two dynamics: 
the iteration cost of multiplicative t\^atonnement dynamics scale up at a $\tfrac{1}{2}$ rate of that of relative t\^atonnement, 
which also matches our theoretical results as $\rho \rightarrow \infty$ and thus $\sigma \rightarrow 1$: 
multiplicative t\^atonnement has a dependency $(d_{\min})^{-1} + (d_{\min})^{-\sigma}$ and relative t\^atonnement has a dependency $(d_{\min})^{-1-\sigma}$. 

\section*{Acknowledgments}
The research of Bhaskar Ray Chaudhury was supported by NSF CAREER Grant CCF-2441580.
The research of Christian Kroer was supported by the Office of Naval Research awards N00014-22-1-2530 and N00014-23-1-2374, and the National Science Foundation awards IIS-2147361 and IIS-2238960.
The research of Ruta Mehta was supported by NSF Grant CCF-2334461.



\bibliographystyle{plain}
\bibliography{refs}

@article{chaudhury2025t,
  title={Tâtonnement Dynamics for Fisher Markets with Chores},
  author={Chaudhury, Bhaskar Ray and Kroer, Christian and Mehta, Ruta and Nan, Tianlong},
  journal={arXiv preprint arXiv:2511.21162},
  year={2025}
}

@inproceedings{chaudhury2024competitive,
  title={Competitive Equilibrium for Chores: from Dual Eisenberg-Gale to a Fast, Greedy, LP-based Algorithm},
  author={Chaudhury, Bhaskar Ray and Kroer, Christian and Mehta, Ruta and Nan, Tianlong},
  booktitle={Proceedings of the 25th ACM Conference on Economics and Computation},
  pages={40--40},
  year={2024}
}

@article{lu2018relatively,
  title={Relatively smooth convex optimization by first-order methods, and applications},
  author={Lu, Haihao and Freund, Robert M and Nesterov, Yurii},
  journal={SIAM Journal on Optimization},
  volume={28},
  number={1},
  pages={333--354},
  year={2018},
  publisher={SIAM}
}

@article{ding2024stochastic,
  title={Stochastic Bregman Subgradient Methods for Nonsmooth Nonconvex Optimization Problems},
  author={Ding, Kuangyu and Toh, Kim-Chuan},
  journal={arXiv preprint arXiv:2404.17386},
  year={2024}
}

@article{bauschke2017descent,
  title={A descent lemma beyond Lipschitz gradient continuity: first-order methods revisited and applications},
  author={Bauschke, Heinz H and Bolte, J{\'e}r{\^o}me and Teboulle, Marc},
  journal={Mathematics of Operations Research},
  volume={42},
  number={2},
  pages={330--348},
  year={2017},
  publisher={Informs}
}

@article{goldman2015spliddit,
  title={Spliddit: Unleashing fair division algorithms},
  author={Goldman, Jonathan and Procaccia, Ariel D},
  journal={ACM SIGecom Exchanges},
  volume={13},
  number={2},
  pages={41--46},
  year={2015},
  publisher={ACM New York, NY, USA}
}

@article{arrow1958stability,
  title={On the stability of the competitive equilibrium, I},
  author={Arrow, Kenneth J and Hurwicz, Leonid},
  journal={Econometrica: Journal of the Econometric Society},
  pages={522--552},
  year={1958},
  publisher={JSTOR},
  url={https://www.jstor.org/stable/1907515}
}

@article{samuelson1941stability,
  title={The stability of equilibrium: comparative statics and dynamics},
  author={Samuelson, Paul A},
  journal={Econometrica: Journal of the Econometric Society},
  pages={97--120},
  year={1941},
  publisher={JSTOR},
  doi={10.2307/1906872}
}

@article{bogomolnaia2017competitive,
  title={Competitive division of a mixed manna},
  author={Bogomolnaia, Anna and Moulin, Herv{\'e} and Sandomirskiy, Fedor and Yanovskaya, Elena},
  journal={Econometrica},
  volume={85},
  number={6},
  pages={1847--1871},
  year={2017},
  publisher={Wiley Online Library},
  doi={10.3982/ECTA14564}
}

@inproceedings{garg2020computing,
  title={Computing competitive equilibria with mixed manna},
  author={Garg, Jugal and McGlaughlin, Peter},
  booktitle={AAMAS Conference proceedings},
  year={2020}
}

@article{devanur2008market,
	title={Market equilibrium via a primal-dual algorithm for a convex program},
	author={Devanur, Nikhil R and Papadimitriou, Christos H and Saberi, Amin and Vazirani, Vijay V},
	journal={Journal of the ACM (JACM)},
	volume={55},
	number={5},
	pages={1--18},
	year={2008},
	publisher={ACM New York, NY, USA},
    doi={10.1145/1411509.1411512}
}

@article{cole2019balancing,
	title={Balancing the Robustness and Convergence of Tatonnement},
	author={Cole, Richard and Tao, Yixin},
	journal={arXiv preprint arXiv:1908.00844},
	year={2019},
    doi={https://doi.org/10.48550/arXiv.1908.00844}
}

@article{zhang2011proportional,
	title={Proportional response dynamics in the Fisher market},
	author={Zhang, Li},
	journal={Theoretical Computer Science},
	volume={412},
	number={24},
	pages={2691--2698},
	year={2011},
	publisher={Elsevier},
    doi={https://doi-org.ezproxy.cul.columbia.edu/10.1016/j.tcs.2010.06.021}
}

@article{eisenberg1959consensus,
  title={Consensus of subjective probabilities: The pari-mutuel method},
  author={Eisenberg, Edmund and Gale, David},
  journal={The Annals of Mathematical Statistics},
  volume={30},
  number={1},
  pages={165--168},
  year={1959},
  publisher={JSTOR},
  url={https://www.jstor.org/stable/2237130}
}

@article{hayek1945use,
  title={The use of knowledge in society},
  author={Hayek, Friedrich August},
  journal={The American economic review},
  volume={35},
  number={4},
  pages={519--530},
  year={1945},
  publisher={JSTOR}
}

@article{eaves1975finite,
  title={A finite algorithm for the linear exchange model},
  author={Eaves, B Curtis},
  year={1975},
  doi={10.1016/0304-4068(76)90028-8}
}

@article{nenakov1983one,
  title={One algorithm for finding solutions of the Arrow-Debreu model},
  author={Nenakov, EI and Primak, ME},
  journal={Kibernetica},
  volume={3},
  pages={127--128},
  year={1983}
}

@inproceedings{birnbaum2011distributed,
  title={Distributed algorithms via gradient descent for fisher markets},
  author={Birnbaum, Benjamin and Devanur, Nikhil R and Xiao, Lin},
  booktitle={Proceedings of the 12th ACM conference on Electronic commerce},
  pages={127--136},
  year={2011},
  organization={ACM},
  doi={10.1145/1993574.1993594}
}

@article{shmyrev2009algorithm,
  title={An algorithm for finding equilibrium in the linear exchange model with fixed budgets},
  author={Shmyrev, Vadim I},
  journal={Journal of Applied and Industrial Mathematics},
  volume={3},
  number={4},
  pages={505},
  year={2009},
  publisher={Springer},
  doi={https://doi.org/10.1134/S1990478909040097}
}

@inproceedings{branzei2021proportional,
  title={Proportional dynamics in exchange economies},
  author={Br{\^a}nzei, Simina and Devanur, Nikhil and Rabani, Yuval},
  booktitle={Proceedings of the 22nd ACM Conference on Economics and Computation},
  pages={180--201},
  year={2021},
  doi={10.1145/3465456.3467644}
}

@article{cheung2019tatonnement,
  title={Tatonnement beyond gross substitutes? Gradient descent to the rescue},
  author={Cheung, Yun Kuen and Cole, Richard and Devanur, Nikhil R},
  journal={Games and Economic Behavior},
  year={2019},
  publisher={Elsevier},
  doi={10.1016/j.geb.2019.03.014}
}

@inproceedings{Orlin10,
  author       = {James B. Orlin},
  title        = {Improved algorithms for computing fisher's market clearing prices:
                  computing fisher's market clearing prices},
  booktitle    = {{STOC}},
  pages        = {291--300},
  publisher    = {{ACM}},
  year         = {2010},
  doi={https://doi.org/10.1145/1806689.1806731}
}

@inproceedings{cole2008fast,
  title={Fast-converging tatonnement algorithms for one-time and ongoing market problems},
  author={Cole, Richard and Fleischer, Lisa},
  booktitle={Proceedings of the fortieth annual ACM symposium on Theory of computing},
  pages={315--324},
  year={2008},
  doi={10.1145/1374376.1374422}
}

@article{BranzeiS24,
  author       = {Simina Br{\^{a}}nzei and
                  Fedor Sandomirskiy},
  title        = {Algorithms for Competitive Division of Chores},
  journal      = {Math. Oper. Res.},
  volume       = {49},
  number       = {1},
  pages        = {398--429},
  year         = {2024},
  doi          = {10.1287/moor.2023.1361}
}

@inproceedings{ChaudhuryGMM22,
  author       = {Bhaskar Ray Chaudhury and
                  Jugal Garg and
                  Peter McGlaughlin and
                  Ruta Mehta},
  title        = {Competitive Equilibrium with Chores: Combinatorial Algorithm and Hardness},
  booktitle    = {{EC}},
  pages        = {1106--1107},
  publisher    = {{ACM}},
  year         = {2022},
  doi          = {10.1145/3490486.3538255}
}

@inproceedings{BoodaghiansCM22,
  author       = {Shant Boodaghians and
                  Bhaskar Ray Chaudhury and
                  Ruta Mehta},
  title        = {Polynomial Time Algorithms to Find an Approximate Competitive Equilibrium
                  for Chores},
  booktitle    = {{SODA}},
  pages        = {2285--2302},
  publisher    = {{SIAM}},
  year         = {2022},
  doi          = {10.1137/1.9781611977073.92}
}

@article{duan2015combinatorial,
  title={A combinatorial polynomial algorithm for the linear Arrow--Debreu market},
  author={Duan, Ran and Mehlhorn, Kurt},
  journal={Information and Computation},
  volume={243},
  pages={112--132},
  year={2015},
  publisher={Elsevier},
  doi={10.1016/j.ic.2014.12.009}
}

@article{ye2008path,
  title={A path to the Arrow--Debreu competitive market equilibrium},
  author={Ye, Yinyu},
  journal={Mathematical Programming},
  volume={111},
  number={1},
  pages={315--348},
  year={2008},
  publisher={Springer},
  doi={https://doi.org/10.1007/s10107-006-0065-5}
}

@inproceedings{gao2020first,
  author    = {Yuan Gao and Christian Kroer},
  title     = {First-Order Methods for Large-Scale Market Equilibrium Computation},
  booktitle = {Neural Information Processing Systems 2020, NeurIPS 2020},
  year      = {2020},
  url={https://dl.acm.org/doi/10.5555/3495724.3497548}
}

@article{jain2007polynomial,
  title={A polynomial time algorithm for computing an Arrow--Debreu market equilibrium for linear utilities},
  author={Jain, Kamal},
  journal={SIAM Journal on Computing},
  volume={37},
  number={1},
  pages={303--318},
  year={2007},
  publisher={SIAM},
  doi={10.1137/S0097539705447384}
}

@inproceedings{cheung2013tatonnement,
  title={Tatonnement beyond gross substitutes? Gradient descent to the rescue},
  author={Cheung, Yun Kuen and Cole, Richard and Devanur, Nikhil},
  booktitle={Proceedings of the forty-fifth annual ACM symposium on Theory of computing},
  pages={191--200},
  year={2013}
}

@inproceedings{goktas2023tatonnement,
  author={Goktas, Denizalp and Zhao, Jiayi and Greenwald, Amy},
  title={T\^{a}tonnement in Homothetic Fisher Markets},
  year={2023},
  isbn={9798400701047},
  booktitle={Proceedings of the 24th ACM Conference on Economics and Computation},
  pages={760-781},
  numpages={22},
  series={EC'23},
  doi={10.1145/3580507.3597746}
}

@inproceedings{goktas2021consumer,
  title={A Consumer-Theoretic Characterization of Fisher Market Equilibria},
  author={Goktas, Denizalp and Viqueira, Enrique Areyan and Greenwald, Amy},
  booktitle={International Conference on Web and Internet Economics},
  pages={334--351},
  year={2021},
  organization={Springer},
  doi={10.1007/978-3-030-94676-0_19}
}

@inproceedings{nan2023fast,
  title={Fast and interpretable dynamics for fisher markets via block-coordinate updates},
  author={Nan, Tianlong and Gao, Yuan and Kroer, Christian},
  booktitle={Proceedings of the AAAI Conference on Artificial Intelligence},
  volume={37},
  number={5},
  pages={5832--5840},
  year={2023}
}

@inproceedings{cheung2018dynamics,
  title={Dynamics of distributed updating in fisher markets},
  author={Cheung, Yun Kuen and Cole, Richard and Tao, Yixin},
  booktitle={Proceedings of the 2018 ACM Conference on Economics and Computation},
  pages={351--368},
  year={2018}
}

@article{arrow1959stability,
  title={On the stability of the competitive equilibrium, II},
  author={Arrow, Kenneth J and Block, Henry D and Hurwicz, Leonid},
  journal={Econometrica: Journal of the Econometric Society},
  pages={82--109},
  year={1959},
  publisher={JSTOR}
}

@inproceedings{codenotti2005market,
  title={Market equilibrium via the excess demand function},
  author={Codenotti, Bruno and McCune, Benton and Varadarajan, Kasturi},
  booktitle={Proceedings of the thirty-seventh annual ACM symposium on Theory of computing},
  pages={74--83},
  year={2005},
  doi={10.1145/1060590.1060601}
}

@inproceedings{cheung2012tatonnement,
  title={Tatonnement in ongoing markets of complementary goods},
  author={Cheung, Yun Kuen and Cole, Richard and Rastogi, Ashish},
  booktitle={Proceedings of the 13th ACM Conference on Electronic Commerce},
  pages={337--354},
  year={2012},
  doi={10.1145/2229012.2229039}
}

@article{walras1874elements,
  title={El{\'e}ments d’{\'e}conomie pure},
  author={Walras, L{\'e}on},
  journal={Economica},
  year={1874}
}

@inproceedings{nan2025convergence,
  title={On the convergence of t{\^a}tonnement for linear fisher markets},
  author={Nan, Tianlong and Gao, Yuan and Kroer, Christian},
  booktitle={Proceedings of the AAAI Conference on Artificial Intelligence},
  volume={39},
  number={13},
  pages={14027--14035},
  year={2025},
  doi={10.1609/aaai.v39i13.33535}
}

@inproceedings{kakade2004graphical,
  title={Graphical economics},
  author={Kakade, Sham M and Kearns, Michael and Ortiz, Luis E},
  booktitle={International Conference on Computational Learning Theory},
  pages={17--32},
  year={2004},
  organization={Springer},
  doi={https://doi.org/10.1007/978-3-540-27819-1_2}
}

@inproceedings{wu2007proportional,
  title={Proportional response dynamics leads to market equilibrium},
  author={Wu, Fang and Zhang, Li},
  booktitle={Proceedings of the thirty-ninth annual ACM symposium on Theory of computing},
  pages={354--363},
  year={2007},
  doi={https://doi.org/10.1145/1250790.1250844}
}

@inproceedings{andrade2021graphical,
  title={Graphical economies with resale},
  author={Andrade, Gabriel and Frongillo, Rafael and Srinivasan, Sharadha and Gorokhovsky, Elliot},
  booktitle={Proceedings of the 22nd ACM Conference on Economics and Computation},
  pages={71--90},
  year={2021},
  doi={https://doi.org/10.1145/3465456.3467628}
}

@inproceedings{zhao2023fisher,
  title={Fisher markets with social influence},
  author={Zhao, Jiayi and Goktas, Denizalp and Greenwald, Amy},
  booktitle={Proceedings of the AAAI Conference on Artificial Intelligence},
  volume={37},
  number={5},
  pages={5900--5909},
  year={2023},
  doi={https://doi.org/10.1609/aaai.v37i5.25731}
}

@article{tao2025fisher,
  title={Fisher Meets Lindahl: A Unified Duality Framework for Market Equilibrium},
  author={Tao, Yixin and Zheng, Weiqiang},
  journal={arXiv preprint arXiv:2511.04572},
  year={2025}
}

@inproceedings{DBLP:conf/aldt/MatteiW13,
  author       = {Nicholas Mattei and
                  Toby Walsh},
  editor       = {Patrice Perny and
                  Marc Pirlot and
                  Alexis Tsouki{\`{a}}s},
  title        = {PrefLib: {A} Library for Preferences http://www.preflib.org},
  booktitle    = {Algorithmic Decision Theory - Third International Conference, {ADT}
                  2013, Bruxelles, Belgium, November 12-14, 2013, Proceedings},
  series       = {Lecture Notes in Computer Science},
  pages        = {259--270},
  publisher    = {Springer},
  year         = {2013},
  url          = {https://doi.org/10.1007/978-3-642-41575-3\_20},
  doi          = {10.1007/978-3-642-41575-3\_20},
  bibsource    = {dblp computer science bibliography, https://dblp.org}
}

@article{zhang2025second,
  title={The Second-Order T$\backslash$\^{} atonnement: Decentralized Interior-Point Methods for Market Equilibrium},
  author={Zhang, Chuwen and He, Chang and Jiang, Bo and Ye, Yinyu},
  journal={arXiv preprint arXiv:2508.04822},
  year={2025}
}

@article{chen2025accelerated,
  title={Accelerated Price Adjustment for Fisher Markets with Exact Recovery of Competitive Equilibrium},
  author={Chen, He and Jiang, Chonghe and So, Anthony Man-Cho},
  journal={arXiv preprint arXiv:2510.07759},
  year={2025}
}

@article{cheung2025proportional,
  title={Proportional Response Dynamics in Gross Substitutes Markets},
  author={Cheung, Yun Kuen and Cole, Richard and Tao, Yixin},
  journal={arXiv preprint arXiv:2506.02852},
  year={2025}
}

@inproceedings{DevanurK08,
  author    = {Nikhil R. Devanur and Ravi Kannan},
  title     = {Market Equilibria in Polynomial Time for Fixed Number of
               Goods or Agents},
  booktitle = {49th Annual {IEEE} Symposium on Foundations of Computer
               Science ({FOCS})},
  pages     = {45--53},
  year      = {2008},
  publisher = {{IEEE} Computer Society},
  doi       = {10.1109/FOCS.2008.30}
}

@inproceedings{JainVY05,
  author    = {Kamal Jain and Vijay V. Vazirani and Yinyu Ye},
  title     = {Market Equilibria for Homothetic, Quasi-Concave Utilities
               and Economies of Scale in Production},
  booktitle = {Proceedings of the Sixteenth Annual {ACM-SIAM} Symposium on
               Discrete Algorithms ({SODA})},
  pages     = {63--71},
  year      = {2005},
  publisher = {{SIAM}}
}

\newpage

\appendix

\begin{center}
    \Large \textbf{Appendix}
\end{center}

\section{Additional related work}

\textbf{Computation of CE in Fisher markets.} A series of foundational works established convex programming and linear complementarity problem (LCP) characterizations of the CE set across various economic models~\cite{eisenberg1959consensus,nenakov1983one,eaves1975finite}. Building on these formulations, \cite{Orlin10} obtained the first strongly polynomial-time algorithm in this setting. Subsequent research has produced a diverse body of algorithms for computing CE in broader market models---such as exchange and Arrow-Debreu markets—going well beyond the linear Fisher case~\cite{gao2020first,kakade2004graphical,andrade2021graphical,zhao2023fisher}. 
These include interior point methods~\cite{jain2007polynomial,ye2008path},
as well as combinatorial approaches~\cite{duan2015combinatorial}.
For \emph{CES markets} specifically, polynomial-time computation has been
established under various restrictions, including weak gross substitutes
regimes ($\sigma > 1$)~\cite{codenotti2005market,cheung2019tatonnement}, fixed numbers of agents or goods~\cite{DevanurK08}, and more general homothetic utilities via convex programming~\cite{JainVY05}.
In the chores market, a separate line of work has yielded polynomial-time algorithms achieving arbitrarily good approximations of CE~\cite{BoodaghiansCM22,ChaudhuryGMM22,chaudhury2024competitive}. 
More recently, \cite{BranzeiS24} gave a polynomial-time algorithm for exactly computing a CE in linear chores Fisher markets whenever the number of agents or chores is constant. 

\section{\texorpdfstring{Proofs of~\cref{sec:pre}}{Proofs of sec:pre}}
\label{app:sec:chores-FM}

\begin{proposition}
    Given a chores Fisher market instance $\mathcal{M} = (n, m, (d_i)_{i=1}^n, (s_j)_{j=1}^m, (B_i)_{i=1}^n)$, 
    there is another instance $\tilde{\mathcal{M}}$ that satisfies~\cref{assump:wlog}, and there is a bijection between the sets of CE of $\mathcal{M}$ and $\tilde{\mathcal{M}}$.
    \label{prop:wlog}
\end{proposition}
\begin{proof}[Proof of~\cref{prop:wlog}]
Given a chores Fisher market instance $\mathcal{M} = (n, m, (d_i)_{i=1}^n, (s_j)_{j=1}^m, (B_i)_{i=1}^n)$, 
there is another instance $\tilde{\mathcal{M}} = (n, m, (\tilde{d}_i)_{i=1}^n, (\tilde{s}_j)_{j=1}^m, (\tilde{B}_i)_{i=1}^n)$ where 
\begin{equation*}
    \tilde{s}_j := 1, \qquad \tilde{B}_i := \frac{B_i}{\sum_{i'=1}^n B_{i'}}, \qquad \tilde{d}_i(\bfx_i) := \frac{d_i(\bfs \odot \bfx_i)}{d_i(\mathbf{s})}.
\end{equation*}
One can verify that each $\tilde{d}_i$ is CCH, and $\tilde{d}_i(\bfx_i) = 0$ if and only if $\bfx_i = \mathbf{0}_m$.
Moreover, $\tilde{\mathcal{M}}$ satisfies~\cref{assump:wlog}.
Next, we claim that 
$(\bfp, \bfx)$ is a CE of $\mathcal{M}$ if and only if $(\tilde{\bfp}, \tilde{\bfx})$ is a CE of $\tilde{\mathcal{M}}$ where 
\begin{equation*}
    (\tilde{\bfp}, \tilde{\bfx}) := \left(\frac{\bfs \odot \bfp}{\sum_{i = 1}^n B_i},  \bfs^{-1} \odot \bfx \right), \quad \bfs^{-1} = (s_1^{-1}, \ldots, s_m^{-1}), 
\end{equation*}
which is true because of the following correspondences.
\begin{align*}
    \tilde{\bfx}_i \in \underset{{\bfx'_i \in \mathbb{R}^m_+}}{\textnormal{argmin}} \left\{  \tilde{d}_i(\bfx'_i) \,\big\vert\, \inp*{\tilde{\bfp}}{\bfx'_i} \geq \tilde{B}_i \right\} \, &\Leftrightarrow \, \bfs^{-1} \odot \bfx \in \underset{{\bfx'_i \in \mathbb{R}^m_+}}{\textnormal{argmin}} \left\{  d_i(\bfs \odot \bfx'_i) \mid \inp*{\bfs \odot \bfp}{\bfx'_i} \geq B_i \right\} \\ 
    &\Leftrightarrow \, \bfs^{-1} \odot \bfx \in \underset{{\bfs^{-1} \odot \bfx''_i \in \mathbb{R}^m_+}}{\textnormal{argmin}} \left\{  d_i(\bfx''_i) \mid \inp*{\bfp}{\bfx''_i} \geq B_i \right\} \\ 
    &\Leftrightarrow \, \bfx \in \underset{\bfx''_i \in \mathbb{R}^m_+}{\textnormal{argmin}} \left\{  d_i(\bfx''_i) \mid \inp*{\bfp}{\bfx''_i} \geq B_i \right\}.
\end{align*}
For any $j \in [m]$, 
\begin{equation*}
    \sum_{i=1}^n \tilde{x}_{ij} \geq 1 \Leftrightarrow \sum_{i=1}^n x_{ij} \geq s_j \quad 
    \tilde{p}_j \left(\sum_{i=1}^n \tilde{x}_{ij} - 1 \right) \Leftrightarrow \frac{p_j}{\sum_{i=1}^n B_i} \left(\sum_{i=1}^n x_{ij} - s_j\right).
    \qedhere
\end{equation*}
\end{proof}

\section{\texorpdfstring{Proofs in~\cref{sec:multiplicative-tatonnement-convergence}}{Proofs in sec:multiplicative-tatonnement-convergence}}
\label{app:sec:convergence}

\begin{repeatproposition}{\ref{prop:sum-p-zeta-equal-0}}[Walras's law]
    If $\bfp \in \Delta_B$, then we have 
    $\inp{\bfp}{\bfy} = 0$ for all $\bfy \in Y(\bfp)$. 
\end{repeatproposition}
\begin{proof}[Proof of~\cref{prop:sum-p-zeta-equal-0}]
    Let $\bfy \in Y(\bfp)$ be any excess supply vector, and
    let $\bfb \in \mathbb{R}^{n \times m}_{+}$ be the earnings for each agent $i$ and chore $j$ corresponding to $\bfy$. 
    That is, by definition $y_j = 1 -\frac{\sum_{i=1}^n b_{ij}}{p_j}$ for all $j \in [m]$.
    By agent's optimality, we have $\sum_{j=1}^m b_{ij} = B_i$ for all $i \in [n]$. 
    Then, since $\bfp \in \Delta_B$, we have 
    \begin{equation*}
        \inp{\bfp}{\bfy} = 
        \sum_{j = 1}^m p_j \left( 1 - \frac{\sum_{i = 1}^n b_{ij}}{p_j} \right) = \sum_{j = 1}^m p_j - \sum_{i = 1}^n \sum_{j = 1}^m b_{ij} = \sum_{j = 1}^m p_j - \sum_{i = 1}^n B_i = 0.
        \qedhere
    \end{equation*} 
\end{proof}

We repeat the definition of $\ell_0$ (for all disutility functions satisfying~\cref{assump:strictly-increasing}) here: 
\begin{equation*}
    \ell_0 = \frac{1}{3m} \min_{i} \left\{ \underset{j, j' \in [m]}{\min_{\tilde{\bfx}_i \geq 0, \bfd'_i > 0, }} \left\{ \frac{d'_{ij}}{d'_{ij'}}\ \left| \ {j, j' \in [m]},  \tilde{x}_{ij} \geq  \frac{1}{2 n m}, \tilde{x}_{ij'} \leq 2m B_i, \bfd'_i \in \partial d_i(\tilde{\bfx}_i) \right. \right\} \right\} > 0.
\end{equation*}

For CES disutility function with $\rho \in (1, \infty)$, i.e., $d_i(\bfx_i) = (\sum_{j=1}^m (d_{ij} x_{ij})^\rho)^{1/\rho}$, then one can show that 
\begin{equation}
    \ell_0 = \frac{1}{3m} \min_{i \in [n]} \left\{ \min_{j, j' \in [m]} \left(\frac{d_{ij}}{d_{ij'}}\right)^\rho \left( \frac{\norm*{\bfB}_1}{4 n m^2 B_i} \right)^{\rho - 1} \right\} > 0. 
    \label{eq:nui defn CES}
\end{equation}
The quantity $\ell_0$ is useful in the proof as the excess supply is positive whenever a price is lower than $\ell_0$, which is shown by~\cite[Lemma 3]{chaudhury2025t}.
\begin{lemma}
    In a chores Fisher market with CCH disutilities satisfying~\cref{assump:strictly-increasing}, given any $\bfp \in \mathbb{R}^m_+$ such that $\max_{j \in [m]} p_{j} \geq \frac{\norm*{\bfB}_1}{2m}$, for any chore $j$, if $p_j \leq \ell_0$ then 
    $y_j > 1 - \frac{1}{2m} \geq \frac{1}{2}$ for all $\bfy \in Y(\bfy)$.
    \label{lem:excess-demand-barrier}
\end{lemma}


\begin{repeatlemma}{\ref{lem:eta-y-bound}}[Boundedness of excess supply]
    In a CCH chores Fisher market that satisfies~\cref{assump:strictly-increasing}, denote $d_{\min}^{(i)} := \min_{j \in [m]} d_i(\mathbf{e}_j) > 0$ where $\mathbf{e}_j$ is the basis vector with a $1$ in coordinate $j$ and $0$ elsewhere. 
    Then, for any $\bfp \in \Delta_\bfB$, 
    $1 - (d_{\min})^{-1} \leq y_{j} \leq 1$ for any $j \in [m]$ and $\bfy \in Y(\bfp)$, where $d_{\min} := \min_i d_{\min}^{(i)}$.
    If $d_i(\cdot)$ is a convex CES disutility function for each agent $i$, then $1 - (\min_{i, j} d_{ij})^{-1} \leq y_{j} \leq 1$ for any $j \in [m]$ and $\bfy \in Y(\bfp)$.
\end{repeatlemma}
\begin{proof}[Proof of~\cref{lem:eta-y-bound}]
    \label{proof:lem:eta-y-bound}
    For each agent $i \in [n]$ and any 
    \begin{equation}
        \bfx^*_{i} \in \argmin_{\bfx \geq 0} \big\{ d_i(\bfx_i) \, \big\vert \, \sum\nolimits_{j = 1}^m p_j x_{ij} \geq B_i \big\}, 
        \label{eq:disutility-minimization}
    \end{equation}
    it follows that 
    \begin{equation}
        d_i(\bfx^*_i) \leq d_i\left(\frac{B_i}{\norm{\bfp}_1}\cdot \mathbf{1}_m\right) = d_i\left( B_i\cdot \mathbf{1}_m \right) = B_i, 
        \label{eq:bounded-disutility-CE}
    \end{equation}
    where the first inequality holds because $({B_i}/{\norm{\bfp}_1}) \cdot \mathbf{1}_m$ is a feasible solution to the minimization problem in~\cref{eq:disutility-minimization}, and 
    the inequalities follows by $\bfp \in \Delta_B$ and $\norm{\bfB}_1 = 1, d_i(\mathbf{1}_m) = 1$ by~\cref{assump:wlog}.     

    Since $d_i$ is increasing in each coordinate, 
    by~\cref{eq:bounded-disutility-CE}, 
    we have 
    \[
       x_{ij}^* d_{\min}^{(i)} \leq x_{ij}^* d_i(\mathbf{e}_j) =
       d_i(x_{ij}^* \mathbf{e}_j) \leq d_i(\mathbf{x}_i^*) \leq B_i  \ \forall\, j \in [m], 
    \]
    where the first equality is because $d_i$ is $1$-homogeneous, and the second inequality is because of~\cref{assump:strictly-increasing}.
    Thus, we have $x^*_{ij} \leq {B_i}/{d_{\min}^{(i)}}$.
    
    Therefore, for any $\bfy \in Y(\bfp)$, we have $1 \geq y_j = 1 - \sum_{i = 1}^n x^*_{ij} \geq 1 - \sum_{i = 1}^n {B_i}/{d_{\min}^{(i)}}$.
    Applying a uniform bound for $d^{(i)}_{\min}$ over all $i \in [n]$, we have 
    \begin{equation*}
        1 - \sum_{i = 1}^n \frac{B_i}{d_{\min}^{(i)}} \geq 1 - \frac{1}{\min_i d_{\min}^{(i)}} \sum_{i = 1}^n B_i = 1 - \frac{1}{d_{\min}}.
    \end{equation*}

    For CES disutility function with $\rho \in (1, \infty)$, i.e., $d_i(\bfx_i) = (\sum_{j=1}^m (d_{ij} x_{ij})^\rho)^{1/\rho}$, 
    \begin{equation}
        d^{(i)}_{\min} = \min_j d_{ij} \textnormal{ and } d_{\min} = \min_{i, j} d_{ij}.
        \label{eq:CES-excess-supply-boundness}
    \end{equation}
    This completes the proof.
\end{proof}


Leveraging this sharper boundedness of the excess demand, we can show that, the desired behaviors of the discrete-time multiplicative t\^atonnement holds with a larger stepsize than a na\"ive bound.
Moreover, this admissible stepsize will not blow up when $\rho \rightarrow \infty$.

\begin{repeatlemma}{\ref{lem:discrete-time-multiplicative-tatonnement-properties}}[Strict positiveness of prices]
    Let $\beta := 1 + (d_{\min})^{-1} > 0$.
    Let $\{ \bfp^k \}_{k \geq 0}$ be a sequence of iterates generated by~\eqref{alg:mul-ttm} with the initial point $\bfp^0 \in \textnormal{relint}\left( \Delta_B \right)$ and $\eta^k \leq \frac{1}{2\beta}$ for all $k$. 
    Then, we have 
    $(i)$ $\bfp^k \in \Delta_B$ for all $k \geq 0$, 
    $(ii)$ $p^{k+1}_j > p^k_j + \frac{1}{2}\eta^k$ if $p^k_j \leq \ell_0$ for any $k \geq 0$ and $j \in [m]$, 
    $(iii)$ if additionally $\sum_{k=0}^\infty \eta^k = \infty$, then there exists a finite index $k_0 \geq 0$ such that $p^k_j \geq \frac{\ell_0}{2}\; \forall\, j \in [m]$ for all $k \geq k_0$.
\end{repeatlemma}
\begin{proof}[Proof of~\cref{lem:discrete-time-multiplicative-tatonnement-properties}]
    The above lemma follows from the following three facts. 
    Informally, these points create a discrete-time threshold that prevents any price from approaching zero, provided the step size is chosen small enough.
    \begin{enumerate}
        \item For any $\bfp^k \in \Delta_B$, $\bfp^{k+1} \in \Delta_B$; 
        \item For any $\bfp^k \in \Delta_B$ and each $j \in [m]$, if $p^k_j \leq \ell_0$ then $p^{k+1}_j > p^k_j + \frac{1}{2} \eta^k$; 
        \item For any $\bfp^k \in \RR^m_+$ and each $j \in [m]$, if $p^k_j > \ell_0$, then $p^{k + 1}_j > \frac{\ell_0}{2}$. 
    \end{enumerate}  
    The first fact is true because of~\cref{prop:sum-p-zeta-equal-0}.
    The second fact follows from~\cref{lem:excess-demand-barrier}. 
    The third fact is true because of the following: for any $\bfy^k \in Y(\bfp^k)$, we have $| y^k_j | \leq \beta$. 
    Thus, $p^{k + 1}_j = p^k_j(1 + \eta^k y^k_j) > p^k_j (1 - \eta^k \beta) \geq \frac{1}{2} p^k_j > \frac{\ell_0}{2}$. 
    
    Then, we prove~\cref{lem:discrete-time-multiplicative-tatonnement-properties}. 
    The first statement in~\cref{lem:discrete-time-multiplicative-tatonnement-properties} follows from the first fact by induction, and second statement was proved as the second fact above.
    If there exists a set of chores $J$ such that $p^0_j \in [0, \ell_0)$ for any $j \in J$. 
    Because $p^{k + 1}_j > p^k_j + \frac{1}{2}\eta^k$ if $p^k_j < \ell_0$, $p^k \in \Delta_B$, and $\sum_{k=0}^\infty \eta^k = \infty$, 
    there is a time point $\kappa_j$ for each $j \in J$ such that it exceeds the barrier $\ell_0$ for the first time, i.e., $\kappa_j := \min\{ k \mid p^k_j > \ell_0 \}$.
    The above facts then imply that $p^k_j \geq \frac{\ell_0}{2}$ for all $k \geq \kappa_j$. 
    Hence, letting $k_0 := \max_{j \in J} \kappa_j$, we show the third statement in~\cref{lem:discrete-time-multiplicative-tatonnement-properties} is correct.
\end{proof}

We restate~\cref{prop:ding-proposition} here for reference.
\begin{repeatproposition}{\ref{prop:ding-proposition}}
    Let $\{ \bfp^k \}_{k \geq 0}$ be a sequence of iterates generated by Bregman subgradient update described in~\cref{eq:bregman-subgradient-update-1,eq:bregman-subgradient-update-2}, and suppose that the process satisfies that 
    $(1)$ the sequences $\{ \bfp^k \}_{k \geq 0}$, $\{ \log{(\bfp^k)} \}_{k \geq 0}$, and $\{ \bfy^k \}_{k \geq 0}$ are uniformly bounded almost surely,
    $(2)$ the stepsizes satisfies $\sum_{k=0}^\infty \eta^k = \infty$, and either $\eta^k = o(\frac{1}{\log{k}})$ or $\sum_{k=0}^\infty (\eta^k)^2 < \infty$, 
    $(3)$ the perturbation error satisfies $\lim_{k \rightarrow \infty} \nu^k \rightarrow 0$, 
    $(4)$ there is a potential function $f$ such that $-\bfy^k \in \partial f(\bfp^k)$\footnote{$\partial f(\bfp^k)$ denotes Clarke differential of function $f$ at $\bfp^k$.}, and $f$ is lower bounded, 
    $(5)$ the critical value set $\{ f(\bfp) \mid \mathbf{0}_m \in \partial f(\bfp) \}$ has empty interior in $\mathbb{R}$. 
    Then, almost surely, any cluster point of $\{ \bfp^k \}_{k \geq 0}$ is a critical point and the function values $\{ f(\bfp^k) \}_{k \geq 0}$ converge.
\end{repeatproposition}

\begin{repeattheorem}{\ref{thm:asymptotic-convergence}}[Asymptotic convergence]
    Let $\{ \eta^k \}_{k \geq 0}$ be a sequence of stepsizes satisfying (i) $0 < \eta^k \leq 1/(2\beta)$, (ii) $\sum_{k = 0}^\infty \eta^k = \infty$, and (iii) $\eta^k = o(\frac{1}{\log{k}})$ or $\sum_{k=0}^\infty (\eta^k)^2 < \infty$.
    Let $\{ \bfp^k \}_{k \geq 0}$ be a sequence of iterates generated by~\eqref{alg:mul-ttm} with $\{ \eta^k \}_{k \geq 0}$ and any initial point $\bfp^0 \in \textnormal{relint}(\Delta_B)$. 
    Then, every limit point of $\{ \bfp^k \}_{k\geq 0}$ is a CE of the chores Fisher market.
\end{repeattheorem}
\begin{proof}
    By~\cref{lem:excess-demand-barrier}, we show that the sequence of iterates $\left\{ \bfp^k \right\}_{k \geq 0}$ lie in $\textnormal{relint}(\Delta_B)$ 
    and have a positive lower bound, 
    if the stepsizes are properly upper bounded. 

    Combining~\cref{lem:discrete-time-multiplicative-tatonnement-properties} with $\bfp^0 \in \mathbb{R}^m_{++}$, we show a price lower bound for~\eqref{alg:mul-ttm} dynamics under stated small stepsizes. 
    Since $\bfp^k \in \textnormal{relint}(\Delta_B)$ holds, there is a trivial upper bound $\norm{\bfB}_1 = 1$. 
    It follows that the prices and the excess demands are bounded, and therefore \cref{prop:ding-proposition}.(1) is satisfied. Note that, we need the concrete lower bound to make sure the iterates $\{ \nabla h_{\textnormal{KL}}(\bfp^k) \}_{k \geq 1}$ is uniformly bounded almost surely. 
    \cref{prop:ding-proposition}.(2) are true by the stepsize selection rule.
    Note that 
    \begin{align*}
        \norm*{\frac{\nabla h_{\textnormal{KL}}(\bfp^{k + 1}) - \nabla h_{\textnormal{KL}}(\bfp^k)}{\eta^k} - \bfy^k}_2 
        =& \norm*{\frac{\log{(1 + \eta^k \bfy^k)}}{\eta^k} - \bfy^k}_2 \\ 
        =& \sum_{j=1}^m \frac{1}{(\eta^k)^2}\left( \log{(1 + \eta^k y^k_j)} - \eta^k y^k_j \right)^2 \\ 
        \leq& \sum_{j = 1}^m \frac{1}{(\eta^k)^2}\left( \eta^k y^k_j \right)^4 \\ 
        =& \sum_{j = 1}^m (\eta^k)^2 \left( y^k_j \right)^4 \leq \beta^4 m (\eta^k)^2 \rightarrow 0 \ \textnormal{ as } k \rightarrow \infty, 
    \end{align*}
    where the first inequality holds because $|\log(1 + x) - x| \leq x^2$ for all $|x| \leq \frac{1}{2}$ and $| \eta^k y^k_j | \leq \frac{1}{2}$, and the second inequality holds by $| y^k_j | \leq \beta$.
    Thus, \cref{prop:ding-proposition}.(3) is satisfied.

    To see \cref{prop:ding-proposition}.(4), by~\cite[Lemma 6]{chaudhury2025t}, the chores potential function $f$ is a valid potential function satisfies $\partial f(\bfp) = - Y(\bfp)$, and has a lower bound over $\bfp \in \Delta_B$.
    The second one is true because stationary points of $f$ are disconnected and the number of critical points is finite, and hence the critical value set cannot include an open interval. 

    Therefore, the theorem holds by~\cref{prop:ding-proposition}.
\end{proof}

\section{\texorpdfstring{Proofs of~\cref{sec:convergence-rate}}{Proofs of sec:convergence-rate}}
\label{app:sec:convergence-rate}

\begin{repeatlemma}{\ref{lem:relative-smoothness-KL}}[Relative smoothness of $f$ w.r.t. $h_{\mathrm{KL}}$]
For any two price vectors $\bfp, \bfq \in \Delta_B$, we have 
\begin{equation}
    f(\bfp) \leq f(\bfq) + \langle \nabla f(\bfq), \bfp - \bfq \rangle + (\sigma - 1)m^{\sigma - 1}(\min_{i,j} d_{ij})^{-\sigma} \KL(\bfp, \bfq).
\end{equation}
\end{repeatlemma}
\begin{proof}
For any $\bfp \in \Delta_{B}$, for any $i \in [n]$, denote $w_{ij}(\bfp) := \frac{d_{ij}^{-\sigma} p_j^{\sigma}}{\sum_{\ell=1}^m d_{i\ell}^{-\sigma} p_\ell^{\sigma}}$.
By calculation, 
\begin{align*}
    \nabla^2 f(\bfp) 
    =& \sum_{i = 1}^n B_i \diag(\bfp)^{-1} \left( (\sigma - 1) \diag(\mathbf{w}_i(\bfp)) - \sigma \mathbf{w}_i(\bfp) \mathbf{w}_i(\bfp)^\top  \right) \diag(\bfp)^{-1} \\ 
    \preceq& \sum_{i = 1}^n B_i \diag(\bfp)^{-1} \left( (\sigma - 1) \diag(\mathbf{w}_i(\bfp)) \right) \diag(\bfp)^{-1}, 
\end{align*}
since the second term is rank-one and thus positive semidefinite.
To show 
\begin{equation}
    \nabla^2 f(\bfp) \preceq L \nabla^2 h_{\mathrm{KL}}(p) = L \diag(\bfp)^{-1}, 
    \label{eq:twice-differentiable-L}
\end{equation}
it suffices to prove that 
\begin{equation*}
    (\sigma - 1) p_\ell \sum_{i=1}^n B_i\,\frac{w_{i\ell}}{p_\ell^2} \leq L \quad \forall\, \ell \in [m].
\end{equation*}

By Jensen's inequality and the convexity of $x^\sigma$, we have 
\begin{equation*}
    (\frac{1}{m}\sum_{j=1}^m p_j)^\sigma \leq 
    \frac{1}{m} \sum_{j=1}^m p_j^\sigma, 
\end{equation*}
hence 
\begin{equation*}
    \sum_{j=1}^m p_j^\sigma \geq \frac{1}{m^{\sigma - 1}} \sum_{j=1}^m p_j = \frac{1}{m^{\sigma - 1}} \sum_{i=1}^n B_i = \frac{1}{m^{\sigma - 1}}.
\end{equation*}
Then, we have  
\begin{equation}
    \sum_{\ell=1}^m d_{i\ell}^{-\sigma} p_\ell^{\sigma} \geq (\max_j d_{ij})^{-\sigma} \sum_{j = 1}^m p_j^\sigma \geq \frac{1}{m^{\sigma - 1}(\max_j d_{ij})^{\sigma}}.
    \label{eq:sum-lower-bound}
\end{equation}
Since 
\begin{equation*}
    p_j \sum_{i=1}^n B_i\frac{w_{ij}}{p_j^2} = \sum_{i=1}^n B_i\frac{d_{ij}^{-\sigma} p_j^{\sigma - 1}}{\sum_{\ell=1}^m d_{i\ell}^{-\sigma} p_\ell^{\sigma}} \leq  \sum_{i=1}^n B_i m^{\sigma - 1}\left(\frac{\max_j d_{ij}}{\min_j d_{ij}}\right)^{\sigma}, 
\end{equation*}
where the last inequality holds because \cref{eq:sum-lower-bound}, $d_{i\ell} \leq \min_j d_{ij} \ \forall\, \ell$, and $p_j^{\sigma - 1} \leq \norm{\bfp}_1^{\sigma - 1} = \norm{\bfB}_1^{\sigma -1} = 1$.
Thus, we have that \cref{eq:twice-differentiable-L} holds with $L = (\sigma - 1) \sum_{i=1}^n B_i m^{\sigma - 1}(\frac{\max_j d_{ij}}{\min_j d_{ij}})^{\sigma} \leq (\sigma - 1)m^{\sigma - 1}\max_i(\frac{\max_j d_{ij}}{\min_j d_{ij}})^{\sigma}\sum_{i = 1}^n B_i = (\sigma - 1)m^{\sigma - 1}\max_i(\frac{\max_j d_{ij}}{\min_j d_{ij}})^{\sigma}$.

Therefore, we have 
\begin{equation}
    f(\bfp) \leq f(\bfq) + \langle \nabla f(\bfq), \bfp - \bfq \rangle + (\sigma - 1)m^{\sigma - 1}\max_i\left(\frac{\max_j d_{ij}}{\min_j d_{ij}}\right)^{\sigma} \KL(\bfp, \bfq).
    \label{eq:tight-smoothness-KL}
\end{equation}
because $\nabla^2 f(\bfp) \preceq L \nabla^2 h_{\mathrm{KL}}(p)$ is equivalent to relative $L$-smoothness of $f$ with respect to $h_{\mathrm{KL}}$~\citep[Proposition 1.1.]{lu2018relatively}\citep[Proposition 1]{bauschke2017descent}.
Furthermore, 
since $(\max_{j} d_{ij})^\rho \leq \sum_{j = 1}^m d_{ij}^\rho = (d_i(\mathbf{1}_m))^\rho = 1$, we have  
$(\max_j d_{ij})^\sigma \leq 1$.
Combining this uniform bound with~\cref{eq:tight-smoothness-KL}, we completes the proof.
\end{proof}

\begin{repeatlemma}{\ref{lem:descent-inequality}}[Descent property]
    Let $\{ \bfp^k \}_{k \geq 0}$ be a sequence of iterates generated by~\eqref{alg:mul-ttm} with any positive stepsizes $\{ \eta^k \}_{k \geq 0}$ and any initial point $\bfp^0 \in \textnormal{relint}(\Delta_B)$.
    We have that 
    \begin{equation}
        f(\bfp^{k + 1}) \leq f(\bfp^k) - \frac{1}{\eta^k} \left( \KL(\bfp^{k + 1}, \bfp^k) \right) - \inp{\nu^k}{\bfp^{k+1} - \bfp^k} + L \KL(\bfp^{k + 1}, \bfp^k).
        \label{eq:descent-inequality}
    \end{equation}
\end{repeatlemma}
\begin{proof}
    The first order optimality then provides
    \begin{equation*}
        \mathbf{0}_m = \eta^k (\bfy^k - \nu^k) - \left( \nabla h_{\textnormal{KL}}(\bfp^{k + 1}) - \nabla h_{\textnormal{KL}}(\bfp^k) \right).
    \end{equation*}
    and then by the definition of Bregman divergence 
    \begin{align*}
        \eta^k \langle \bfy^k - \nu^k, \bfp^{k + 1} - \bfp^k \rangle 
        =& \langle \nabla h_{\textnormal{KL}}(\bfp^{k + 1}) - \nabla h_{\textnormal{KL}}(\bfp^k), \bfp^{k + 1} - \bfp^k \rangle \\
        =& \KL(\bfp^{k + 1}, \bfp^k) + \KL(\bfp^k, \bfp^{k + 1}).
    \end{align*}

    \begin{align*}
        &f(\bfp^{k + 1}) \\ 
        \leq& f(\bfp^k) + \langle \nabla f(\bfp^k), \bfp^{k + 1} - \bfp^k \rangle + L \KL(\bfp^{k + 1}, \bfp^k) \\ 
        =& f(\bfp^k) - \langle \bfy^k, \bfp^{k + 1} - \bfp^k \rangle + L \KL(\bfp^{k + 1}, \bfp^k) \\ 
        =& f(\bfp^k) - \frac{1}{\eta^k} \left( \KL(\bfp^{k + 1}, \bfp^k) + \KL(\bfp^k, \bfp^{k + 1}) \right) - \inp{\nu^k}{\bfp^{k+1} - \bfp^k} + L \KL(\bfp^{k + 1}, \bfp^k) \\
        \leq& f(\bfp^k) - \frac{1}{\eta^k} \KL(\bfp^{k + 1}, \bfp^k) - \inp{\nu^k}{\bfp^{k+1} - \bfp^k} + L \KL(\bfp^{k + 1}, \bfp^k), 
    \end{align*}
    where we use the relative smoothness in the first inequality.
\end{proof}

\begin{repeatlemma}{\ref{lem:link-distances}}
    If $\eta^k \leq \frac{1}{2\beta}$ for all $k \geq 0$, then we have $- \inp{\nu^k}{\bfp^{k+1} - \bfp^k} \leq \beta(\eta^k)^2 \norm{\bfy^k}_{\bfp^k}^2$and $\KL(\bfp^{k + 1}, \bfp^k) \geq \frac{(\eta^k)^2}{3} \norm{\bfy^k}_{\bfp^k}^2$.
\end{repeatlemma}
\begin{proof}
    Since $\eta^k \leq \frac{1}{2\beta}$ thus $| \eta^k y^k_j | \leq \frac{1}{2}$, 
    \begin{align*}
        -\inp{\nu^k}{\bfp^{k+1} - \bfp^k} &= \sum_{j = 1}^m \left(\frac{\log{(1 + \eta^k y^k_j)}}{\eta^k} - y^k_j\right) \eta^k p^k_j y^k_j \\ 
        &\leq \sum_{j = 1}^m \left| \frac{\log{(1 + \eta^k y^k_j)}}{\eta^k} - y^k_j\right| \eta^k p^k_j | y^k_j | \\ 
        &= \sum_{j = 1}^m \left| \log{(1 + \eta^k y^k_j)} - \eta^k y^k_j\right| p^k_j | y^k_j | \\ 
        &\leq \sum_{j = 1}^m (\eta^k y^k_j)^2 p^k_j \beta = \beta(\eta
        ^k)^2\sum_{j=1}^m p^k_j (y^k_j)^2,
    \end{align*}
    where the last inequality holds because $|\log(1 + x) - x| \leq x^2$ for all $|x| \leq \frac{1}{2}$, together with $| \eta^k y^k_j | \leq \frac{1}{2}$ and $| y^k_j | \leq \beta$.


    Also, since $(1+x)\log(1+x)-x \geq \frac{x^2}{3}$ for all $|x| \leq \frac{1}{2}$, we have 
    \begin{align*}
        \KL(\bfp^{k+1}, \bfp^k)
        &= \sum_{j=1}^m p_j^k(1+\eta^k y^k_j)\log(1+\eta^k y^k_j) \\
        &= \eta^k \sum_{j=1}^m p_j^k y_j^k
        + \sum_{j=1}^m p_j^k \big( (1 + \eta^k y^k_j) \log(1 + \eta^k y^k_j) - \eta^k y^k_j \big) \\
        &\geq \frac{1}{3}\sum_{j=1}^m p_j^k \eta^k (y^k_j)^2 = \frac{(\eta^k)^2}{3}\sum_{j=1}^m p_j^k (y_j^k)^2.
        \qedhere
    \end{align*}
\end{proof}

\begin{repeattheorem}{\ref{thm:convergence-rate}}[Convergence rate guarantee]
    Let $\{ \bfp^k \}_{k \geq 0}$ be a sequence of iterates generated by~\eqref{alg:mul-ttm} with any initial point $\bfp^0 \in \textnormal{relint}(\Delta_B)$ and any positive stepsizes $\{ \eta^k \}_{k \geq 0}$ satisfying 
    \begin{equation*}
        \eta^k \equiv \eta \leq \frac{1}{12(L + \beta)}, 
    \end{equation*}
    then we can find an iterate $(\bfp, \bfy)$ satisfies $\norm{\bfy}_{\bfp} \leq \varepsilon$ in $\lceil
    {12(f(\bfp^0) - \underline{f})}/{\eta \varepsilon^2} \rceil$ iterations.
    Here, $L = (\sigma-1)\,m^{\sigma-1}(\min_{i,j} d_{ij})^{-\sigma}$ and $\beta = 1 + (\min_{i,j} d_{ij})^{-1}$.
\end{repeattheorem}

\begin{proof}[Proof of~\cref{thm:convergence-rate}]

    Since $\eta^k = \eta \leq \frac{1}{12(L + \beta)} \leq \min\{ \frac{1}{2L}, \frac{1}{12\beta} \}$, we have 
    \begin{align*}
        f(\bfp^{k + 1}) 
        \leq& f(\bfp^k) - \left( \frac{1}{\eta^k} - L \right)\KL(\bfp^{k+1}, \bfp^k) + \beta (\eta^k)^2 \sum_{j = 1}^m p^k_j (y^k_j)^2 \\ 
        \leq& f(\bfp^k) - \frac{1}{2\eta^k} \KL(\bfp^{k+1}, \bfp^k) + \frac{\eta^k}{12} \sum_{j = 1}^m p^k_j (y^k_j)^2 \\ 
        \leq& f(\bfp^k) - \frac{\eta^k}{6} \sum_{j = 1}^m p^k_j (y^k_j)^2 + \frac{\eta^k}{12} \sum_{j = 1}^m p^k_j (y^k_j)^2 \\ 
        \leq& f(\bfp^k) - \frac{\eta^k}{12} \sum_{j = 1}^m p^k_j (y^k_j)^2.
    \end{align*}
    Equivalently, 
    \begin{equation*}
        T \cdot \left( \min_{0 \leq k \leq T - 1} \sum_{j = 1}^m p^k_j (y^k_j)^2 \right) \leq \sum_{k = 0}^{T - 1} \sum_{j = 1}^m p^k_j (y^k_j)^2 
        \leq 
        \frac{12(f(\bfp^0) - f(\bfp^{T}))}{\eta} \leq \frac{12(f(\bfp^0) - \underline{f})}{\eta}, 
    \end{equation*}
    where $\bfp^0$ is the initial price vector and $\underline{f}$ is the minimum value of $f$ over $\Delta_B$.
    To obtain that $\min_{0 \leq k \leq T - 1} \sum_{j = 1}^m p^k_j (y^k_j)^2 \leq \varepsilon^2$, we need $T \geq \lceil
    \frac{12(f(\bfp^0) - \underline{f})}{\eta \varepsilon^2} \rceil$ iterations.
\end{proof}

\subsection{Details on the rate of relative t\^atonnement}
\label{app:subsec:comparison-relative-tatonnement}

Here, we focus on the case where each convex CES disutility functions are with $\rho > 2$.
Note that we use a different parameterization ($d_i(\bfx_i) = (\sum_{j = 1}^m (d_{ij} x_{ij})^\rho)^{1/\rho}$) for the convex CES disutility functions from that in~\cite{chaudhury2025t} ($d_i(\bfx_i) = (\sum_{j = 1}^m d_{ij} x_{ij}^\rho)^{1/\rho}$). 

We consider the terms with $1/\ell_0$ and lower bound all other nonnegative terms in the expression of the smoothness modulus.
In their Lemma 15, $r_i$ can equal $1$ in our setting (by~\cref{assump:wlog}), then we have 
\begin{align*}
    L \geq \sum_{i=1}^n B_i R_i L_i 
    &\geq \sum_{i=1}^n B_i R_i ((\sigma - 1)\frac{(\max_j d_{ij}^\rho)^{\frac{(\sigma - 1)^2}{\sigma}}}{(\min_j d_{ij}^\rho)^{\sigma - 1}}) \frac{1}{2\ell_0} \\ 
    &= \sum_{i=1}^n B_i (\min_j d_{ij}^\rho)^{-1/\rho} ((\sigma - 1)\frac{(\max_j d_{ij}^\rho)^{\frac{(\sigma - 1)^2}{\sigma}}}{(\min_j d_{ij}^\rho)^{\sigma - 1}}) \frac{1}{2\ell_0}.
\end{align*}

Since we have $d_i(\mathbf{1}_m) = 1$ by~\cref{assump:wlog}, we have 
\begin{equation*}
    \max_{j} d_{ij}^\rho \geq \frac{1}{m}\sum_{j = 1}^m d_{ij}^\rho = \frac{1}{m}.
\end{equation*}
Therefore, 
\begin{equation*}
    L \geq \sum_{i=1}^n B_i (\min_j d_{ij})^{-1} ((\sigma - 1)\frac{1}{m^{\frac{(\sigma - 1)^2}{\sigma}}(\min_j d_{ij}^\rho)^{\sigma - 1}}) \frac{1}{2\ell_0}.
\end{equation*}
Since $\rho (\sigma - 1) = \rho (\frac{\rho}{\rho - 1} - 1) = \rho \frac{1}{\rho - 1} = \sigma$, and $\frac{(\sigma - 1)^2}{\sigma} = \frac{1}{\rho(\rho - 1)} \leq \frac{1}{2}$, we have 
\begin{equation*}
    L \geq (\sigma - 1) m^{-\frac{1}{2}} (\min_j d_{ij})^{-(\sigma + 1)} \frac{1}{2\ell_0}.
\end{equation*}
Since we can take $\eta = \frac{1}{L}$, and their convergence guarantee leads to $\lceil 2L(f(\bfp^0) - \underline{f}) / \varepsilon^2 \rceil$ iterations required to compute an $\varepsilon$-approximate CE.


The following lemma shows that $f(\bfp^0) - \underline{f}$ is always bounded by a number that is polynomial in problem parameters.
\begin{proposition}[Boundedness of $f$]
    For convex CES disutility function with $\rho > 1$, i.e., $d_i(\bfx_i) = (\sum_{j = 1}^m (d_{ij} x_{ij})^\rho)^{1/\rho}$, 
    $ \log{\frac{1}{e \max_{i} \norm{\bfd_i}_\rho}} \leq f(\bfp) \leq  \log{\frac{1}{e \min_{i, j} d_{ij}}}$ for any $\bfp \in \Delta_B$.
    \label{prop:lipschitz}
\end{proposition}
\begin{proof}
    By solving $\min_{\bfp \geq 0, \norm{\bfp}_1 = C} \sum_{j = 1}^m d_{ij}^{-\sigma} p_j^\sigma$, we have 
    \begin{equation}
        \sum_{j = 1}^m d_{ij}^{-\sigma} p_j^\sigma \geq \left( \frac{\norm{\bfp}_1}{\norm{\bfd_i}_\rho} \right)^\sigma, 
        \label{eq:min-inner-log}
    \end{equation}
    therefore $f(\bfp)$ is differentiable when $\bfp$ is away from $\mathbf{0}_m$. 
    By~\cref{eq:min-inner-log}, we have 
    \begin{equation*}
        f(\bfp) \geq - \sum_{j = 1}^m p_j + \sum_{i = 1}^n \frac{B_i}{\sigma} \log{\left(\left( \frac{\norm{\bfp}_1}{\norm{\bfd_i}_\rho} \right)^\sigma\right)} \geq \norm{\bfB}_1 \log{\frac{\norm{\bfB}_1}{e \max_{i} \norm{\bfd_i}_\rho}}.
    \end{equation*}
    Similarly, by solving $\max_{\bfp \geq 0, \norm{\bfp}_1 = C} \sum_{j = 1}^m d_{ij}^{-\sigma} p_j^\sigma$, we have 
    \begin{equation}
        \sum_{j = 1}^m d_{ij}^{-\sigma} p_j^\sigma \leq \left( \frac{\norm{\bfp}_1}{\min_{j} d_{ij}} \right)^\sigma, 
        \label{eq:max-inner-log}
    \end{equation}
    thus 
    \begin{equation*}
        f(\bfp) \leq - \sum_{j = 1}^m p_j + \sum_{i = 1}^n \frac{B_i}{\sigma} \log{\left(\left( \frac{\norm{\bfp}_1}{\min_{j} d_{ij}} \right)^\sigma\right)} \leq \norm{\bfB}_1 \log{\frac{\norm{\bfB}_1}{e \min_{i, j} d_{ij}}}.
        \qedhere
    \end{equation*}
\end{proof}

\newpage

\section{Additional numerical experiments}

\begin{figure}[h]
    \centering
    \includegraphics[width=1\linewidth]{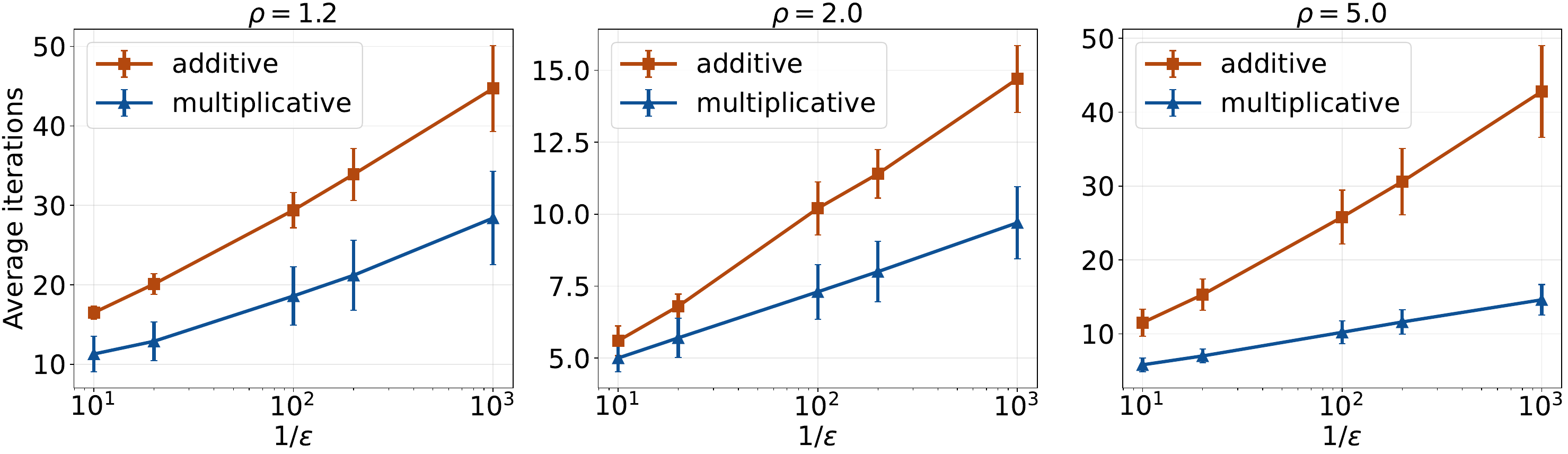}
    
    \caption{Multiplicative T\^atonnement v.s. Relative T\^atonnement on AAMAS bidding instances.
    The error bar corresponds to $1$ standard deviation.}
    \label{fig:bidding}
\end{figure}

\begin{figure}[h]
    \centering
    \includegraphics[width=1\linewidth]{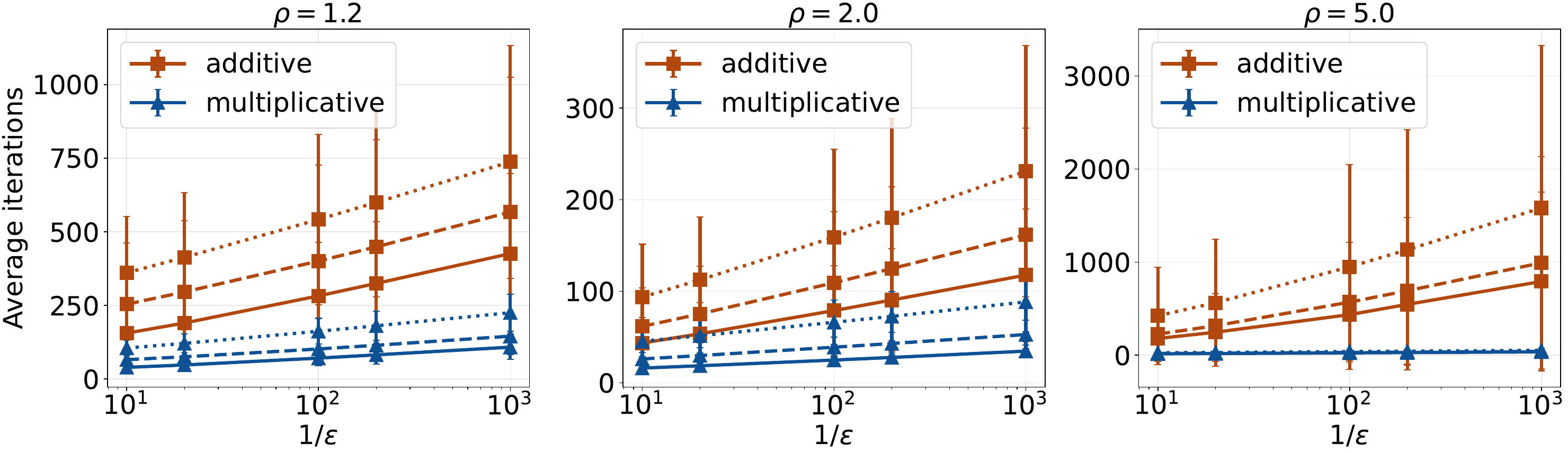}
    \caption{Multiplicative T\^atonnement v.s. Relative T\^atonnement on synthetic instances with lognormal value distributions.
    Solid lines correspond to $100 \times 500$ instances, dashed lines correspond to $100 \times 1000$ instances, dotted lines correspond to $100 \times 2000$ instances.
    The error bar corresponds to $0.1$ standard deviation.}
    \label{fig:synthetic-lognormal}
\end{figure}

\end{document}